\documentclass[twocolumn]{autart}    

\usepackage{multirow}
\usepackage{amsmath,amssymb,mathtools}

\usepackage{tabularx}  
\usepackage{ragged2e}  
\usepackage{gensymb}
\usepackage{graphics}
\usepackage{epstopdf}
\usepackage{color}
\usepackage{latexsym}
\usepackage{float}
\usepackage[title]{appendix}%
\usepackage{xcolor}%
\usepackage{textcomp}%
\usepackage{manyfoot}%
\usepackage{booktabs}%
\usepackage{algorithm}%
\usepackage{algorithmicx}%
\usepackage{algpseudocode}%
\usepackage{listings}%
\usepackage{newcmd}
\usepackage{caption}
\usepackage{graphicx}      
\usepackage{tabu}
\usepackage{bbm}
\usepackage{MnSymbol}
\usepackage{bm}
\usepackage{xfrac}
\usepackage{textcomp}
\usepackage{cite}
\usepackage{hyperref}
\hypersetup{urlcolor=green, colorlinks=true, linkcolor=cyan, citecolor=red, filecolor=black}
\definecolor{orange}{RGB}{255,127,0}

\usepackage{natbib}
\usepackage[margin=2mm]{subfig}

\usepackage{tikz}

\newtheorem{theorem}{Theorem}
\newtheorem{proposition}[theorem]{Proposition}%
\newtheorem{lemma}{Lemma}
\newtheorem{corollary}{Corollary}
\newtheorem{remark}{Remark}%

\newenvironment{proof}{\noindent\textbf{Proof.} }{\hfill$\blacksquare$\par}
\begin{document}
\begin{frontmatter}

\title{Finite-time Stable Pose Estimation on $\Ta\SE$ using Point Cloud and Velocity Sensors} 

\thanks[footnoteinfo]{This paper was not presented at any IFAC 
meeting. Corresponding author A.~K.~Sanyal. Tel. +1 315-443-0466.}

\author{Nazanin S. Hashkavaei}\ead{nsafaeih@syr.edu},
\author{Abhijit Dongare}\ead{abhijit.dongare2892@gmail.com},    
\author{Neon Srinivasu}\ead{neonsrin@syr.edu},
\author[footnoteinfo]{Amit K. Sanyal}\ead{aksanyal@syr.edu}  

\address{Department of Mechanical \& Aerospace Engineering, Syracuse University, Syracuse, NY 13244}        

\begin{keyword}                           
finite-time stability, nonlinear estimation, variational integrator, point cloud sensors, Lyapunov analysis, model-free estimation, Lie algebra, unmanned vehicle.               
\end{keyword}                             

\begin{abstract}                          
This work presents a finite-time stable pose estimator (FTS-PE) for rigid bodies undergoing rotational and translational motion in three dimensions, using measurements from onboard sensors that provide position vectors to inertially-fixed points and body velocities. The FTS-PE is a full-state observer for the pose (position and orientation) and velocities and is obtained through a Lyapunov analysis that shows its stability in finite time and its robustness to bounded measurement noise. Further, this observer is designed directly on the state space, the tangent bundle of the Lie group of rigid body motions, SE(3), without using local coordinates or (dual) quaternion representations. Therefore, it can estimate arbitrary rigid body motions without encountering singularities or the unwinding phenomenon and be readily applied to autonomous vehicles. A version of this observer that does not need translational velocity measurements and uses only point clouds and angular velocity measurements from rate gyros, is also obtained. It is discretized using the framework of geometric mechanics for numerical and experimental implementations. The numerical simulations compare the FTS-PE with a dual-quaternion extended Kalman filter and our previously developed variational pose estimator (VPE). The experimental results are obtained using point cloud images and rate gyro measurements obtained from a Zed 2i stereo depth camera sensor. These results validate the stability and robustness of the FTS-PE.
\end{abstract}

\end{frontmatter}

\section{Introduction}
 State estimation on the Lie group of rigid body motions is essential for aerial vehicles, spacecraft, and underwater vehicles. In the absence of global navigation satellite systems (GNSS) in these applications, onboard sensors that include vision, lidar, infrared, sonar, and inertial sensors are used to provide information on the pose (position and orientation) of the vehicle body with respect to an inertial (spatial) reference frame. 
Pose estimation in real time is required for feedback control of translational and rotational motion during the operations of autonomous 
vehicles. The set of all possible poses of a rigid body is given by the special Euclidean group in three dimensions, denoted as $\SE$ in ~\cite{bloch03,bullo2019geometric}. This group is a semi-direct 
product of the real Euclidean linear space of translations $\bR^3$ and the Lie group of rotations represented by $3 \times 3$ real orthogonal matrices with determinant $1$, denoted  $\SO$. The non-contractible and compact configuration space of three-dimensional rotations, $\SO$, makes $\SE$ a non-contractible space,
and therefore, attitude and pose estimation are inherently nonlinear estimation problems. As the translational and rotational motions of rigid bodies are usually coupled and their attitude cannot be directly measured by onboard sensors, a pose estimation scheme has to compute the pose from vector measurements obtained from sensors mounted on the rigid body. 


Attitude estimation, which is a component of pose estimation, has a well-established history. Early research such as \cite{black1964passive,wahba1965least} designed static attitude determination methods from a set of direction vector 
measurements measured at an instant. However, the performance of static attitude determination schemes can be unsatisfactory in the presence of measurement noise and bias components. Therefore, estimation schemes like the 
extended Kalman filter (\cite{shuster1990kalman,quat2006Chou}) and multiplicative extended Kalman filter (\cite{markley1988attitude}) are often used for attitude estimation. Nevertheless, 
these schemes may perform poorly due to the restrictive assumptions they make about measurement noise, process noise, and state transitions; see \cite{crassidis2007survey,ICRA15}. Other prior 
approaches focused on developing attitude estimation schemes on the Lie group $\SO$, e.g.,~\cite{Bonnabel2009sym,mahony2008complimentary,silvestre2010pose,markley2006so3,sanyal2006optimal,lageman2010gradient,Moutinho2015}. These 
approaches avoid the use of unit quaternions or local coordinate descriptions like Euler angles to represent attitude. Local coordinates encounter singularities in representing attitude globally, while attitude estimators using unit quaternions may exhibit unwinding-type instability because 
antipodal unit quaternions on the hypersphere $\bS^3$ correspond to a single attitude. This instability can also result in longer convergence times compared to stable schemes. 
More recent work on nonlinear estimation schemes on $\SO$ and $\SE$ include \cite{mahony2017geometric,hamel2020deterministic,VANGOOR2021,RezaAutomatica21,wang2022observability,mahony2022observer,Barbon2023,mathavaraj2024consensus,chababo24,fogoalmawe24,shaaban2025attitude,wang2025pose,mousse2025efficient,shaaban2025position}.

Despite the recent advances in estimators designed on Lie groups, homogeneous and symmetric spaces, the majority of these designs do not explicitly account for the stability of the estimator and its domain of convergence.
Due to the geometry of the compact manifold $\SO$, no continuous attitude observer can provide convergence of the attitude estimation error to identity from all initial attitude and 
angular velocity estimation errors as shown by~\cite{RezaAutomatica21}. 
A continuous attitude observer or controller can, at best, be {\em almost global} in its region of attraction, as shown in~\cite{bhat2000topological, chaturvedi2011rigid}. For an attitude observer, almost global stability implies that the attitude estimate
converges to the true attitude for almost all initial attitude estimates, except those in a set of zero measure in the state space. The pose estimation scheme presented in \cite{izadi2016rigid} follows
the variational framework of the estimation scheme reported in \cite{Izadi2014rigid}. Moreover, it exhibits almost global asymptotic stability, similar to the variational attitude estimator given in \cite{Izadi2014rigid}.
	
 There are advantages to having finite time stable pose estimation schemes: they have been shown to be more robust to disturbances and noise and provide faster convergence than an asymptotically stable scheme under similar initial conditions~\cite{bhat:2000:finite}. Additionally, a finite time stable estimation 
 scheme automatically satisfies a ``separation principle" in the case estimated state variables are used for feedback control, as observer estimates converge to true states in finite time. Finite time stable (FTS) estimation schemes using 
 sliding mode controllers and neural networks are proposed by \cite{li2015state,zou2011finite}, which are not continuous. Our prior work includes \cite{Bohn2014observer} and \cite{sanyal2014observer}, which proposed almost globally FTS attitude observers. 
 However, they assumed knowledge of the dynamics model including inertia parameters, and did not consider bias in angular velocity measurements. The scheme proposed by \cite{warier2019finite} enables FTS estimation of 
 attitude and unbiased angular velocity measurements on $\Ta\SO$, while \cite{RezaAutomatica21} also estimates an unknown bias in angular velocity measurements, both without requiring knowledge of a dynamics model.
 
This work builds on our prior work on FTS attitude estimation for FTS pose estimation, presenting a novel pose estimation scheme referred to as the Finite-Time Stable Pose Estimator (FTS-PE). Like the FTS attitude 
estimator (FTS-AE) in \cite{RezaAutomatica21}, the FTS-PE uses a Morse-Lyapunov function of estimation errors on the state space to prove its finite-time stability and robustness to measurement noise. But unlike 
the FTS-AE: (a) the state space for the FTS-PE is the tangent bundle $\Ta\SE\simeq\SE\times\se$, which considers coupled rotational and translational motions; and (b) the attitude estimate in the FTS-PE is constructed from pairwise 
combinations of position vectors of (a subset of) points in a three-dimensional (3D) point cloud  measured by an onboard sensor. On the other hand, the FTS-AE scheme can estimate attitude using only inertial vector measurements (e.g., from magnetometer and accelerometer) and/or direction vector measurements (e.g., from a sun sensor).
The key features of the FTS-PE are: (1) it provides pose (both attitude and position) estimates directly on the Lie group $\SE$ of rigid body motions, avoiding singularities or unwinding; (2) it does not require a dynamics model, including knowledge of mass/inertia parameters or a 
measurement noise model; (3) it is uniformly continuous and has almost global finite-time stable (AGFTS) convergence to the true pose and velocities in the absence of measurement errors; and (4) it is robust to time-varying noise in 
measured velocities. These properties are theoretically shown, and validated through numerical simulations that also compare its performance with the variational pose estimator (VPE)~\cite{izadi2016rigid} and the dual quaternion-based pose estimator
in \cite{filipe2015extended}. An experiment using a depth camera sensor with an inbuilt inertial measurement unit (IMU) that gives point cloud and angular velocity measurements, is also used to verify its performance. 

The paper is structured as follows. Section \ref{sec:math} presents the mathematical notations and concepts used in the paper. In Section \ref{sec:Problem}, we introduce the formulation for the problem of pose determination from vector measurements obtained from optical sensors. Section \ref{sec:preliminary} discusses the estimation error dynamics and kinematics. We then generalize Wahba's cost function by choosing a symmetric weight matrix and show that the resulting cost function is a Morse function on the Lie group of rigid-body rotations under some easily satisfied conditions on the weight matrix. This cost function is taken as the potential function for rotational motion, and we also provide some useful lemmas associated with the potential functions for rotations and translations. In Section \ref{sec: ftd}, we present the problem of finite-time stable pose estimation in real-time, present the pose estimation scheme, and prove the finite-time stability of the scheme in the absence of measurement
errors, using nonlinear stability analysis.
Section \ref{sec: robustness} presents an in-depth analysis of the robustness of the FTS pose estimation scheme to measurement errors. Section \ref{sec:sim} provides numerical simulation results of the proposed finite-time stable pose estimator in the absence and presence of 
measurement noise, and its comparison with a few other existing pose estimation schemes. In addition, this section provides results on an experiment implementing the FTS-PE on the ZED 2i stereo depth camera sensor. Finally, we conclude the paper with a summary of results and possible future directions in Section \ref{sec: conc}.

\section{Mathematical Preliminaries}\label{sec:math}
	The field of real numbers is denoted by $\mathbb{R}$, while $\mathbb{R}^n$ and $\mathbb{R}^{n\times m}$ denote the set of real n-dimensional column vectors and real $n\times m$ matrices, respectively. $\mathbb{N}$ denotes the set of natural numbers. 
	The set of all possible attitudes of a rigid body is the special orthogonal group $\SO$ (\cite{murray2017mathematical}), which is defined by:  
	\begin{equation*}
	\SO = \left\{ R \in \bR^{3 \times 3} | R\T R = RR\T = I, \: \det(R) =1 \right\}.
	\end{equation*}
	This is a matrix Lie group under matrix multiplication. The Lie algebra (tangent space at identity) of $\SO$ is denoted $\so$ 
    and identified with the set of $3\times 3$ skew-symmetric matrices:
	\begin{equation*}
	\ensuremath{\mathfrak{so}(3)}=\big\{\mathrm{S}\in\bR^{3 \times 3} \;|\; \mathrm{S}=-\mathrm{S} \T\big\},\; \mathrm{S}= \bbm 0 & -\ms_3 & \ms_2\\ \ms_3 & 0 & -\ms_1\\ -\ms_2 & \ms_1 & 0\ebm.
	\end{equation*}
	Let $(\cdot)^{\times}: \bR^3 \to \so$ denote the bijective map from three-dimensional real Euclidean space $\bR^3$ to $\ensuremath{\mathfrak{so}(3)}$. For a vector $\ms=[\ms_1 \; \ms_2\; \ms_3]\T \in \bR^3$, the matrix $\ms^{\times}$ represents the vector cross product operator, that is $\ms \times r = \ms^{\times}r$, where $r\in \bR^3$; this makes $(\cdot)^{\times}$ a vector space isomorphism. 
	The inverse of $(\cdot)^{\times}$ is denoted by $\vex(\cdot): \so \to \bR^3 $, such that $\vex(\ms^{\times}) = \ms$, for all $\ms^{\times} \in \ensuremath{\mathfrak{so}(3)}$. 
 
 We define the trace inner product on $\bR^{m \times n}$ $\left\langle \cdot, \cdot\right\rangle $ as, $$\left\langle A_1,A_2 \right\rangle = \tr (A_1\T A_2).$$  
Any square matrix $A \in \bR^{n \times n}$ can be written as the sum of unique symmetric and skew-symmetric matrices as given below:
		\begin{align}
		&A = \symf(A) + \skewf(A), \label{eq:SymSkewSum} 
		\end{align}
		where the symmetric and skew-symmetric components are defined as:
		\begin{align}
		\symf(A) = \frac{1}{2} (A+ A \T ),\; &
		\skewf(A) = \frac{1}{2} (A - A \T) \label{eq:SkewDef}.
		\end{align}
		Additionally, if $A_1 \in \mathbb{R}^{n\times n}$ is a symmetric matrix and $A_2 \in \mathbb{R}^{n\times n}$ is a skew-symmetric matrix, then, 
		\begin{align}
		\left\langle A_1, A_2 \right\rangle=& \;0. \label{eq:sym}
		\end{align}
		Thus, the symmetric and skew matrices are orthogonal under the trace inner product. For all $a_1, a_2\in \bR^3$,
		\begin{align}
		\left\langle a_1^{\times}, a_2^{\times} \right\rangle=& \; 2 a_1\T a_2. \label{eq:skew}
		\end{align} 

  The Lie group of rigid body motions in three spatial dimensions is the special Euclidean group, denoted $\SE$. It is the semi-direct product of $\SO$ with $\bR^3$, with $\bR^3$ as the normal subgroup, i.e., $\SE=\SO\ltimes\bR^3$ (\cite{Varadarajan_1984}). If $R\in\SO$ and $b\in\bR^3$, our nominal and adjoint representations of the corresponding element in $\SE$ are, respectively:
  \be
  \mathrm{g} = \begin{bmatrix}
R & b \\
0 & 1
 \end{bmatrix} \in \SE \mbox{ and }
\Ad{\mathrm{g}} =\begin{bmatrix}
R & 0 \\
b^\times R & R
\end{bmatrix}.  \label{SE3repns}
 \ee
The action of $\SE$ on $\bR^3$ is to rotate and translate vectors in $\bR^3$, and is represented as follows:
\begin{align*}
    &\rmg\cdot \mbox{v} = \bbm R & b \\ 0 & 1\ebm \bbm v\\ 1\ebm = \bbm Rv+b \\ 1\ebm, \ \forall\, v\in\bR^3,\, 
  \mbox{ where} \;  \mbox{v} = \bbm v\\ 1\ebm.
\end{align*}
		With these definitions and identities, we formulate the pose estimation problem next.

  \section{Real-time Navigation using 3D Point Cloud Data}\label{sec:Problem}
This section introduces a formulation for the problem of pose determination from vector measurements in the vehicle body-fixed frame $\cB$, along with a velocity measurement model. Our FTS pose estimation scheme using 3D point clouds is based on this formulation, which will be presented in detail in Section \ref{sec: ftd}. 
  \subsection{Pose Measurement Model}
		Let $\mathcal{I}$ denote an inertial frame that is spatially fixed and $\mathcal{B}$ denote the body-fixed frame.
  The attitude of the rigid body is denoted by $R \in \SO$, which transforms vectors in the body frame $\mathcal{B}$ to their counterparts in the inertial frame $\mathcal{I}$. The position of the origin of frame $\cB$ expressed in frame $\cI$ gives the position of the body, denoted as $b \in \bR^3$. The pose is given by the frame transformation from frame $\cB$ to frame $\cI$, as given by eq. \eqref{SE3repns}.
  
Consider a set of $j$ points from a point cloud measured at time $t$ with known and fixed positions in frame $\cI$, denoted as $q_j$. These $j$ points generate $ \big(\begin{smallmatrix} j \\ 2 \end{smallmatrix}\big)$ unique relative position vectors, which are the vectors connecting any two of these observed points.
Let $a_i$ be the relative position of the $i$-th stationary point in frame $\cB$. In the absence of measurement noise, we obtain the position vector of this point in frame $\cI$ as
  \begin{equation}\label{position_vec}
      q_{i} = Ra_{i} + b.
  \end{equation}
The measured vectors in the presence of additive noise can be expressed as:
\begin{equation}
    \bar{a}^{m} = R\T (\bar{q} - b) + \bar{\wp},
\end{equation}
where $\bar{q}$ 
and $\bar{a}^{m}$ are defined as follows: 
\begin{align}
\bar{q} = \dfrac{1}{j}\sum^{j}_{i=1} q_{i},\; 
\bar{a}^{m} = \dfrac{1}{j}\sum^{j}_{i=1} a^{m}_{i}, 
\label{barqbaram-defs}
\end{align}
and $\bar{\wp}$ is the additive measurement noise obtained by averaging the measurement noise vectors for each $a_{i}$.

To estimate the attitude, consider the $ \big(\begin{smallmatrix} j \\ 2 \end{smallmatrix}\big)$ pairwise relative position vectors for the $j$ points selected in the point cloud, denoted $d_{j} = q_{\lambda} - q_{\ell}$ in frame $\cI$. The corresponding vectors in frame $\cB$ as $e_{j} = a_{\lambda} - a_{\ell}$, where $\lambda, \ell$ are any two measured points such that, $\lambda \neq \ell$. 
If the total number of measured vectors, $ \big(\begin{smallmatrix} j \\ 2 \end{smallmatrix}\big) = 2$, then $e_{3} = e_{1} \times e_{2}$ is considered a third measured direction in frame $\cB$ with corresponding vector $d_{3} = d_{1} \times d_{2}$ in frame $\cI$. Therefore,
\begin{equation}\label{D_def}
    d_{j} = Re_{j} \Rightarrow D = RE,
\end{equation}
where  $D = [d_{1} \; \dots \; d_{n} ], \; E = [e_{1} \; \dots \; e_{n} ] \in \bR^{3\times n}$ with $n = 3$ if $ \big(\begin{smallmatrix} j \\ 2 \end{smallmatrix}\big) = 2$ and $n = \big(\begin{smallmatrix} j \\ 2 \end{smallmatrix}\big)$ if $ \big(\begin{smallmatrix} j \\ 2 \end{smallmatrix}\big) > 2$. In the presence of measurement noise, the measured value of matrix $E$ is given by,
\begin{equation}
    E^{m} = R\T D + \mathfrak{L},
\end{equation}
where the columns of matrix $\mathfrak{L} \in \bR^{3\times n}$ are additive noise vectors in the vector measurements made in frame $\cB$.

\subsection{Velocities Measurement Model}
Denote the angular and translational velocity of the rigid
body expressed in frame $\cB$ by $\Omega$ and $\nu$, respectively.
Therefore, the kinematics of the rigid body is
\begin{equation}\label{def_g}
    \dot{R} = R\Omega^{\times}, \quad \dot{b}  = R\nu \; \Rightarrow \dot{\mathrm{g}} = \mathrm{g}\xi^{\vee},
\end{equation}
where $\xi = \begin{bmatrix}
    \small \Omega \\ \nu
\end{bmatrix} \in \bR^{6}$, $\xi^{\vee} = \begin{bmatrix}
    \Omega^{\times} & \nu \\
    0 & 0
\end{bmatrix}$. 
If velocities are directly measured, then the presence of measurement noise, the measured vector of translational and angular velocities, denoted $\xi^m$, is given by
 		\begin{align}
 		\xi^m  = \begin{bmatrix}
    \small \Omega^m \\ \nu^m
\end{bmatrix} = \xi + \delta, \mbox{ where }
 \delta = \begin{bmatrix}\delta_{\Omega}\\ \delta_{\nu}\end{bmatrix}
\label{def_xi_m}
 		\end{align}
 denotes the vector of additive noise in translational and angular velocities. 
		
In the case where only rate gyro measurements of angular velocities are available besides the point cloud measurements, the translational velocity of the rigid body 
can be calculated using these measurements over time, as in \cite{izadi2016rigid}. The rationale for this approach is given by Proposition \ref{implement-prop}.
    
In the following section, we explore the dynamics and kinematics of estimation errors and present preliminary results that contribute to the design of our FTS pose estimation scheme.

\section{Preliminary Concepts for Pose Estimation on $\SE$}\label{sec:preliminary}

The estimated pose and its kinematics are given by
\begin{equation}\label{hat_g}
 \widehat{\mathrm{g}} = \begin{bmatrix}
\widehat{R} & \widehat{b} \\
0 & 1
 \end{bmatrix} \in \SE, \;\; \Dot{\widehat{\mathrm{g}}} = \widehat{\mathrm{g}} \widehat{\xi}^\vee,
\end{equation}
where $\widehat{b}$ is the position estimate, $\widehat{R}$ is the attitude estimate, and $\widehat{\xi}$ is the rigid body velocities estimate, with $\widehat{\mathrm{g}}_0$ as the initial pose estimate. 
The pose estimation error is defined as
\begin{equation}
 h = \mathrm{g}\widehat{\mathrm{g}}^{-1} = \begin{bmatrix}
Q & b - Q\widehat{b} \\
0 & 1
 \end{bmatrix} = \begin{bmatrix}
Q & \chi \\
0 & 1
 \end{bmatrix} \in \SE,
\end{equation}
where $Q = R\widehat{R}\T$ is the attitude estimation error, and $\chi = b - Q\widehat{b}$. In the absence of measurement noise, we have
\begin{equation}\label{varphi}
    \Dot{h} = h \varphi^\vee, \; \mbox{where} \; \varphi(\widehat{\mathrm{g}},\xi^m,\widehat{\xi}) 
= \begin{bmatrix}
\omega   \\
\upsilon 
 \end{bmatrix}
 = \Ad{\widehat{\mathrm{g}}}(\xi - \widehat{\xi}),
\end{equation}
where $\xi^m=\xi \in \bR^{6}$ is the measured rigid body velocities defined as in \eqref{def_xi_m}, $\upsilon$ and $\omega$ are translational and angular velocity estimation errors respectively, and $\Ad{\mathrm{g}}$ as 
defined by eq. \eqref{SE3repns}. The position and attitude estimation error kinematics are:
 \begin{align}
 \begin{split}
     &\Dot{\chi} = Q \upsilon, \quad \Dot{Q} = Q \omega^\times, \\
     \mbox{where } &\omega=\widehat{R}(\Omega - \wh\Omega),\; \upsilon = \widehat{R}(\nu - \wh\nu) - \omega^\times\widehat{b}.
\end{split} \label{err_kine}    
 \end{align}
 			
        The position and the attitude ($\wh{b}$, $\widehat{R}$) estimates are obtained in real time, using the matrices which consist of known inertial vectors $D$ defined in \eqref{D_def}, the corresponding vector measured vectors in the body-fixed frame $E^m$, and the translational and angular velocities measurements $\xi^m$.
        The mass, moment of inertia, and other parameters that occur in the dynamics of the rigid body are unknown. The number of vector measurements can be varying over time, provided that there are at least two non-collinear vectors measured  
        at all times, for unique attitude determination. The design of the FTS pose estimator given in Section \ref{sec: ftd} is shown to provide almost global finite-time stable (AGFTS) estimates of the pose, where these estimates converge to the 
        respective true values $b$ and $R$ in finite time, in the absence of measurement noise. Additionally, the robustness of this estimator in the presence of measurement noise is demonstrated in section \ref{sec: robustness}. 

\subsection{Potential Functions for Pose Estimation Errors}

Consider the following potential function for the position estimation error:
\begin{equation}\label{potential_translation}
    \mathcal{U}_{t}(\wh{g},\bar{a}^{m},\bar{q}) = \dfrac{1}{2} \kappa \;y\T y = \dfrac{1}{2}\kappa \norm{\bar{q} - \wh{R}\bar{a}^{m} - \wh{b}}^{2},
\end{equation}
where $y \equiv y(\wh g,\bar{a}^{m},\bar{q})=  \bar{q} - \wh{R}\bar{a}^{m} - \wh{b}$ is the position estimation errors and $\kappa>0$ is a scalar gain. 
Now, consider the potential function for attitude estimation error. The attitude estimate, denoted by $\widehat{R} \in \SO$, is obtained known inertial vectors $d_1,\ldots,d_n $ and the corresponding vectors $e^m_1,\ldots, e^m_n$ measured in the body frame, where $n$ is 
as defined after eq. \eqref{D_def}. Attitude determination is formulated as the following optimization problem:
\begin{align}
\begin{split}
    &\text{Min}_{\widehat{R}}\, \mathcal{U}_{r}(\widehat{R}, D, E^m)\, \mbox{ where} \\
    \cU_r (\widehat{R}, D, E^m) &=  
    \frac{1}{2} \sum_{i}^{n} w_i (d_i - \widehat{R} e^m_i)\T (d_i - \widehat{R} e_i^m) \\
    &= \frac{1}{2} \left\langle D-\widehat{R}E^m, (D-\widehat{R}E^m) W\right\rangle,
\end{split} \label{eq:wahbacost}
\end{align}
where $w_i>0$ are weight factors. This is referred to as Wahba's problem in the attitude determination and estimation literature, first formulated in \cite{wahba1965least}. 
Here $W = \diag ([w_1,\; w_2,\; \ldots,\; w_n ]) \in \bR^{n\times n}$. The cost function can be generalized such that $W$ is a positive semi-definite matrix satisfying conditions given by 
the following lemmas. Therefore, the total potential function is obtained as a sum of the translational and rotational potential functions and is given by:
\begin{align}
  \begin{split}
    &\mathcal{U}(\wh{g},\bar{a}^{m},\bar{q},D, E^m) = \mathcal{U}_{t}(\wh{g},\bar{a}^{m},\bar{q}) + \mathcal{U}_{r}(\widehat{g}, D, E^m) \\
    &= \dfrac{1}{2} \kappa \;y\T y 
    + \frac{1}{2} \left\langle D-\widehat{R}E^m, (D-\widehat{R}E^m) W\right\rangle,
    \end{split} \label{potential_combined}
\end{align}
where $\mathcal{U}_{t}(\wh{g},\bar{a}^{m},\bar{q}) \; \mbox{and} \; \mathcal{U}_{r}(\widehat{g}, D, E^m)$ are defined by equations \eqref{potential_translation} and \eqref{eq:wahbacost} respectively.
	 The following lemma gives the generalized cost function in the absence of measurement errors.
	\begin{lemma} \label{costU}
		{ Define $Q = R\widehat{R}\T$ as the attitude estimation error.} Let $D \in \bR^{3 \times n}$ be as defined as in \eqref{D_def} with $rank(D) =3$. Let the gain matrix $W$ of the generalization of Wahba's cost function be given by,
		\begin{align}
		W & = D\T \big(DD\T\big)^{-1} K \big(DD\T\big)^{-1} D, \label{eq:Wexp}
		\end{align} 
		where $K = \diag([k_1, k_2, k_3])$ and $k_1 > k_2 > k_3 \geq 1 $. Then, in the absence of measurement errors,  
		\begin{align}\label{Udef}
		\mathcal{U}_r(h, D) &= \frac{1}{2} \left\langle D-\widehat{R}E^m, (D-\widehat{R}E^m) W\right\rangle \nn \\ 
        &= \left\langle K, I - Q \right\rangle,
		\end{align}		
		which is a Morse function on $\SO$ whose critical points are elements of the set,	
		\begin{align}
		\mathcal{C} & = \big\{ I, \diag([-1,-1,1]), \diag([1,-1,-1]), \nonumber \\
		&\quad \quad \diag([-1,1,-1])\big\} \subset \SO. \label{eq:Qset}
	\end{align} 
In addition, $\mathcal{U}_r $ has a global minimum at $Q=I$.
	\end{lemma}	
The proof of this result is given in~\cite{RezaAutomatica21} based on \cite{sanyal2006optimal}, and is not repeated here. 

The instantaneous attitude determination problem can be solved by determining $\widehat{R}$ that minimizes $\mathcal{U}_r (\cdots)$ at any given instant. 
For dynamic pose estimation, substituting \eqref{Udef} in \eqref{potential_combined} gives the total potential function in the absence of measurement errors as follows:
\begin{align}
    \mathcal{U}(h,D,\bar{q}) &= \mathcal{U}_{t}(h,\bar{q}) + \mathcal{U}_{r}(h, D) \nonumber \\
    &= \frac{1}{2} \kappa \;y\T y + \left\langle K, I - Q \right\rangle, \label{totalPF}
\end{align}
where
\begin{align}\label{y_perfect}
 y \equiv y(h,\bar{q}) &= \bar{q} - \wh{R}\bar{a} - \wh{b} \nn\\
 & = Q\T \chi + (I-Q\T)\bar{q},
\end{align}
as $\wh b = Q\T(b - \chi)$.

The following lemma, which relates the error in the attitude estimation to the error in angular velocity estimation, is used in section \ref{sec: ftd} to prove the main result.
\begin{lemma}\label{lemma:ineq}
		Let $K$ be as defined in Lemma \ref{costU}. Then $K=DWD\T$, and in the absence of measurement errors, 
        \be KQ= L\wh{R}\T \mbox{ where } L= DW(E^m)\T. \label{KLrelation} \ee
        Further, in this case the time derivative 
		of $\mathcal{U}_r$ along  trajectories satisfying the kinematic equations \eqref{def_g}, \eqref{def_xi_m} and \eqref{hat_g} is given by
		\begin{align}
		\frac{\di}{\di t}\cU_r(h, D) &= \frac{\di}{\di t}\left\langle K,I-Q\right\rangle = s_K(Q)\cdot \omega \nn \\
        &= s_L(\wh R, D, E^m)\cdot\omega, 
  \label{KIQder} 
		\end{align}
		where 
$s_K(Q) = s_L (\wh R, D, E^m)$ are given by \begin{align}
		&s_K(Q) = \vex(KQ-Q\T K), \mbox{ and }  \label{eq:SKQ} \\
  &s_L (\wh R, D, E^m) = \vex(L {\wh R}\T- \wh R L\T). \label{tilOmSLdefs} 
\end{align}  
	\end{lemma}	
The proof of this lemma is given in~\cite{izadi2016rigid} 
and is omitted here for brevity. 

Therefore, from equations \eqref{err_kine} and \eqref{KIQder}, the time derivative of the total potential function $\cU$ is obtained as
\begin{align}
    \frac{\di}{\di t}\mathcal{U}(h,D,\bar{q}) &= \kappa \;y\T (\upsilon + \omega^{\times}\bar{q}) + s_K (Q)\cdot \omega, \label{cUdot}
\end{align}
which is used in the next section to prove the main result. Hereafter, for the sake of notational convenience, the total potential function is denoted as $\mathcal{U}$.  
\subsection{Useful Prior Results} 
	The lemmas given here are essential to prove the main result on finite time stable pose estimation scheme given in Section \ref{sec: ftd}.
	\begin{lemma} \label{binom}
Let $x$ and $y$ be non-negative real numbers and let $p\in ]1,2[$. Then 
\be x^{(1/p)}+ y^{(1/p)} \ge (x+y)^{(1/p)}. \label{bires} \ee
Moreover, the above inequality is a strict inequality if both $x$ and $y$ are non-zero.
\end{lemma}
The interested reader can find detailed proof of this result in~\cite{Bohn2014observer,bohn2016almost}. For brevity, we omit the proof in this paper.	
\begin{lemma}\label{SKQlemma}
		{ Let $K$ be as defined in Lemma \ref{costU} and $s_K(Q)$ be as given in the equation \eqref{eq:SKQ}. Let $\cS\subset\SO$ be a closed subset containing the identity in its interior, defined by
			\begin{align}
			\cS =& \big\{ Q\in\SO\, :\, Q_{ii}\ge 0 \mbox{ and } Q_{ij}Q_{ji}\le 0 \nn \\
			&\forall i,~ j\in \{1,2,3\},\ i\ne j\big\}. \label{cSdefn}
			\end{align}
			If $Q\in\cS$, then it satisfies
			\be s_K(Q)^T s_K(Q) \ge \tr (K-KQ). \label{sRbound} \ee	}
\end{lemma}
   \vspace{-3mm}
The proof of this result is given in~\cite{bohn2016almost} and is omitted here for brevity. 

The design and stability result of the finite-time stable estimator is given in the following section. Note that the pose estimation error $h=g\widehat{g}^{-1}$ is defined on the special Euclidean group of rigid body motion, $\SE$, which is not a vector space.
Therefore, for stable convergence of the pose estimation error, the total potential function $\cU:\SE\to\bR$ defined by \eqref{totalPF} serves as a Morse function on $\SE$. This forms part of a Morse-Lyapunov function, as defined in the proof in Theorem \ref{FTSestwbias} in Section \ref{sec: ftd}. 
This Morse-Lyapunov function is thereafter shown to guarantee convergence of state estimation errors $(h,\varphi)$ to $(I,0)$ in finite time.

\section{Finite-time Stable Pose Estimation} \label{sec: ftd}
The main result for the proposed finite-time stable pose observer for estimation of rigid body position and orientation is presented here. The uniformly continuous, finite-time stability of the resulting observer is shown using a H\"{o}lder-continuous Morse-Lyapunov function.

\begin{theorem}\label{FTSestwbias}
Consider the pose kinematics, position vectors obtained from the point cloud measurements, and velocity measurements given by equations \eqref{def_g}, \eqref{def_xi_m} and \eqref{hat_g} in the absence of measurement noise. 
Let $p\in]1,2[\,$ and $ \kappa,\, k_p,\, k_\omega, \, k_\upsilon,\, \alpha_1,\, \alpha_2$ be positive observer gains, $L$ and $s_L (\cdots)$ be as defined in Lemma \ref{lemma:ineq}, $y$ be as defined after eq. \eqref{potential_translation}, 
and define the following quantities:
{\allowdisplaybreaks
\begin{align}
    & z_{1}(\widehat{R},D,E^m) = \dfrac{s_{L}}{\Big(s_{L}\T s_{L}\Big)^{1-1/p}}, \label{z1def} \\
    & z_{2}(\wh{g}, \bar{a}^m, \bar{q}) = \dfrac{y}{(y\T y)^{1-1/p}}, \label{z2def} \\
    & \Psi(\widehat{R},D,E^m, \omega) = \omega + \alpha_{1} \;z_{1}(\widehat{R},D,E^m), \label{Psidef} \\
    & \Phi(\wh{g}, \bar{a}^m, \bar{q}, \omega, \upsilon)  = \upsilon + \omega^\times \bar{q}+ \alpha_{2}  \;z_{2}(\wh{g},\bar{a}^m, \bar{q}), \label{Phidef} \\
    & w_{L} (\wh{R},\omega) = \frac{\di}{\di t}s_{L}(\wh{R}) = \vex(L \wh{R}\T \omega^{\times} + \omega^{\times}\wh{R}L\T), 
    \label{wLdef} \\
    & v_{y} = \frac{\di}{\di t}y = \upsilon + \omega^\times (\Bar{q} -y).
    \label{vydef}
\end{align}	}
Now, consider the following observer equations: 
\begin{align}
		\dot\omega &= -k_p s_L 
        - k_{\omega}\frac{\Psi}{\big(\Psi\T\Psi\big)^{1-1/p}} -\frac{\alpha_{1} }{\big(s{_L}\T s_{L}\big)^{1-1/p}}H(s_L) w_L, 
  \label{obs1} \\
        \dot \upsilon &= \bar{q}^\times \dot\omega - k_{p}\; \kappa \;y - k_{\upsilon}\frac{\Phi}{\big(\Phi\T\Phi\big)^{1-1/p}} - \frac{\alpha_{2} }{\big(y\T y\big)^{1-1/p}}H(y) v_y,      \label{obs2} \\
        \wh\xi &= \xi^{m} - \Ad{\wh{\mathrm{g}}^{-1}}\varphi, \label{obs3} \\
         \dot{\wh g} &= \wh g\wh\xi^\vee , \label{obs4}
\end{align}
where $\xi^m$ and $\varphi$ are defined as in \eqref{def_xi_m} and \eqref{varphi}, respectively, the functional dependencies of $s_L$, $w_L$, $v_y$, $z_{1}$, $z_{2}$, $\Psi$ and $\Phi$ have been suppressed for notational convenience, and $H:\bR^3\to \mbox{Sym}(3)$, the space of symmetric $3\times 3$ real matrices, is defined by \be H(x)= I-\frac{2(1-1/p)}{x\T x} xx\T. \label{Hdef} \ee
Then the pose and velocity estimation errors $(h,\varphi)$ converge to $(I,0)\in\SE\times\bR^6$  in a finite time stable manner, from almost all initial conditions except those in a set of measure zero.	
\end{theorem}
	\begin{proof}	
	This result gives the observer equations to estimate the pose and velocities in real time, from the measured quantities $E^m$ and $\xi^m$ in a finite-time stable manner. In the following analysis, all functional dependencies have been dropped for notational convenience. 
	
	Consider the following Morse-Lyapunov function:
\begin{equation}\label{lyapunov}
    \mathcal{V} = k_p\;\mathcal{U} + \dfrac{1}{2}\Psi\T \Psi + \dfrac{1}{2}\Phi\T \Phi,
\end{equation}
where $\cU$ is as defined by eq. \eqref{potential_combined}. Taking the time derivative along the error state trajectories, we get
\begin{align}
    \dot{\mathcal{V}} &= k_p\;\dot{\mathcal{U}} + \Psi\T \dot{\Psi} + \Phi\T \dot{\Phi} \nn \\
    &= k_{p}\left(s_{L}\T \omega  + \kappa \;y\T (\upsilon + \omega^{\times}\bar{q})\right) \label{dot_V} \\ &+ \Psi\T (\dot{\omega} + \alpha_{1}\dot z_1) + \Phi\T (\dot{\upsilon} +\dot\omega^\times\bar{q} + \alpha_{2}\dot z_2), \nn
\end{align}
where $\dot U$ is obtained from Lemma \ref{lemma:ineq} and eq. \eqref{cUdot} $\dot z_1$ and $\dot z_2$ are obtained from time derivatives of \eqref{z1def} and \eqref{z2def}, respectively, and simplified using \eqref{wLdef} and \eqref{Hdef} as:
\begin{align}
    &\dot z_1 = \frac{\di}{\di t}z_{1} = \frac{1}{\big(s{_L}\T 
		s_L\big)^{1-1/p}}H(s_L) w_L \label{dot_z1} \\
&\mbox{and }
    \dot z_2 =\frac{\di}{\di t}z_{2} = \frac{1}{\big(y\T 
		y\big)^{1-1/p}}H(y) v_y.\label{dot_z2}
\end{align}
Additionally, $v_y$ in \eqref{dot_z2} is given by \eqref{vydef} and obtained as follows:
\begin{align*}
v_y = \frac{\di}{\di t}y &=  Q\T \dot{\chi}+ \dot{Q}\T \chi  - \dot{Q}\T \bar{q} \nn \\
&=\upsilon + \omega^{\times}Q\T (\bar{q} -  \chi)\nn \\
&=\upsilon + \omega^{\times}(\bar{q} - y),
\end{align*} 
using the definitions of $\dot{Q}$, $\dot{\chi}$ in \eqref{err_kine}, and $y$ as given in \eqref{y_perfect}.
The expression for $w_L$ in \eqref{dot_z1} given by \eqref{wLdef}, is obtained by applying Lemma \ref{lemma:ineq} and taking the time derivative of $s_K(Q)$ as follows:
\begin{align*}
    w_L &= \frac{\di}{\di t} s_L(\wh R)= \frac{\di}{\di t} s_K(Q) \\
    &= \frac{\di}{\di t}\Big( \vex(KQ-Q\T K) \Big) \\
    &= \vex\big(KQ\omega^\times+ \omega^\times Q\T K\big)= \vex(L \wh{R}\T \omega^{\times} + \omega^{\times}\wh{R}L\T),
\end{align*}
using eq. \eqref{KLrelation} from Lemma \ref{lemma:ineq} in the last step.
%
Now, substituting the equations \eqref{obs1}, \eqref{obs2}, \eqref{dot_z1} and \eqref{dot_z2} into equation \eqref{dot_V} we get
\begin{align}
    \dot{\mathcal{V}} &= k_{p}\left(s_{L}\T \omega  + \kappa\; y\T (\upsilon +   \omega^{\times}\bar{q}) \right) \nonumber \\ 
    &\quad + \Psi\T \biggl(-k_p s_L 
    - k_{\omega}\frac{\Psi}{\big(\Psi\T\Psi\big)^{1-1/p}} \biggl)  \nonumber \\
    &\quad  + \Phi\T \biggl(- k_p\; \kappa\; y  - k_{\upsilon}\frac{\Phi}{\big(\Phi\T\Phi\big)^{1-1/p}}  \biggl) \nonumber \\
    &=  -k_p s_L\T (\Psi - \omega) -  k_p\; \kappa \;y\T (\Phi - \upsilon- \omega^\times\bar{q}) \nonumber \\
    &\quad - k_{\omega}\left(  \Psi\T \Psi \right)^{1/p} - k_{\upsilon}\left(  \Phi\T \Phi \right)^{1/p} \nn \\ 
    &= - k_p \;\alpha_{1}\;s_{L}{\T} z_{1} -  k_p\;\kappa\;\alpha_{2}\; y\T z_{2}\nonumber \\ 
    &\quad - k_{\omega}\left(  \Psi\T \Psi \right)^{1/p} - k_{\upsilon}\left(  \Phi\T \Phi \right)^{1/p} \nn \\
    &= - k_p \left(\alpha_{1}(s_{L}{\T} s_{L})^{1/p} +  \kappa\;\alpha_{2}(y\T y)^{1/p}\right) \nonumber \\
    &\quad - k_{\omega}\left(  \Psi\T \Psi \right)^{1/p} - k_{\upsilon}\left(  \Phi\T \Phi \right)^{1/p}. \label{v_dot_subs}
\end{align}
This time derivative of the Morse-Lyapunov function is negative, except when $s_L$, $y$, $\Psi$ and $\Phi$ are all zero, in which case $\dot\cV=0$. Note that $s_L (\wh R, D, E^m)=s_K(Q)$ according to Lemma \ref{lemma:ineq}. 
Denote the initial estimation errors by $(h_0,\varphi_0) = (Q_0,\chi_0,\omega_0,\upsilon_0)\in\Ta\SE$. Let $S_K(Q_0)$, $\chi_0$, $\omega_0$ and $\upsilon_0$ be not all zero, so that $\dot\cV<0$ initially. Define the 
set $\cX= \{Q\in\cS, \chi\in\bR^3, \omega\in\bR^3, \upsilon\in\bR^3\} \subset\Ta\SE$, where $\cS\subset\SO$ is the closed (hence compact) set containing the identity $I\in\SO$ defined by eq. \eqref{cSdefn} in 
Lemma \ref{SKQlemma}. Now we consider two cases: (1) the initial estimation errors $(Q_0,y_0,\omega_0,\upsilon_0) \notin\cX$; and (2) the initial estimation errors are in the set $\cX$. 

In the first case, if $\cV_0$ is the finite, non-zero initial value of the Morse-Lyapunov function, then this value will decrease until the estimation errors reach the boundary of the set $\cX$ in a 
finite amount of time. When the estimation error $Q$ reaches the neighborhood $\cS$, Lemma \ref{SKQlemma} 
can be applied to obtain the following inequality: 
\begin{equation}\label{sL_U}
    -s_{L}{\T} s_{L} \leq - \mathcal{U}_{r}\left( \wh{g}, D, E^{m}  \right) = - \left\langle K,I-Q\right\rangle.
\end{equation}
Substituting \eqref{sL_U} and \eqref{potential_translation} in \eqref{v_dot_subs}, we get the following inequality for $\dot \cV$ in terms of $\cU_r$, $\cU_t$, $\Psi$ and $\Phi$:
\begin{align}
    \dot{\mathcal{V}} &\leq - \alpha_{1}\, k_p^{1-1/p} (k_p\mathcal{U}_r)^{1/p} -  2^{1/p}\alpha_{2} (k_p\,\kappa)^{1-1/p}(k_p\mathcal{U}_t)^{1/p} \nonumber \\ 
    &\qquad\qquad\qquad - k_{\upsilon}\left(  \Phi\T \Phi \right)^{1/p} - k_{\omega}\left(  \Psi\T \Psi \right)^{1/p} \label{Vdot1} \\
    &\leq - k_{0}\left \{ \left(k_p\;\mathcal{U} \right)^{1/p} +  \left( \dfrac{1}{2}  \Phi\T \Phi \right)^{1/p} + \left(\dfrac{1}{2}\Psi\T \Psi \right)^{1/p} \right \}, \label{Vdot2}
\end{align}
applying Lemma \ref{binom} to the first two terms in the LHS of inequality \eqref{Vdot1}, where
\[k_{0} = \min\left(\alpha_{1}\, k_p^{1-1/p},\, 2^{1/p}\alpha_2(k_p\; \kappa)^{1-1/p},\,2^{1/p}k_{\upsilon},\, 2^{1/p}k_{\omega}  \right),\] and $\mathcal{U}$ is defined in \eqref{potential_combined} as the total potential function.
After applying Lemma \ref{binom} again to the above inequality \eqref{Vdot2}, we obtain:
\begin{align}
    \dot{\mathcal{V}} \leq - k_{0}\left( k_p\;\mathcal{U} + 
 \dfrac{1}{2}\Phi\T \Phi +  \dfrac{1}{2}\Psi\T \Psi \right)^{1/p}  
 \leq - k_{0}\mathcal{V}^{1/p}.\label{fts_Lyp}
\end{align}
Applying inequality \eqref{fts_Lyp}, if $\cV^\cS<\cV_0$ is the value of the Morse-Lyapunov function when the estimation errors reach the boundary of set $\cX$, then the time taken to reach this set is upper bounded by 
\[ T^\cS_0 = \frac{\cV_0-\cV^\cS}{k_0(\cV^\cS)^{1/p}}.  \]
From inequality \eqref{fts_Lyp}, we also conclude that $\cV$ converges to zero in another finite period of time once the attitude estimation error $Q$ reaches the set $\cS$ defined by Lemma \ref{SKQlemma}, and this settling 
time is upper-bounded by $t_s=\frac{1}{k_0}(\cV^S)^{1-1/p}$. Therefore, the total time for the estimation errors to converge to $(h,\varphi)=(I,0)\in\Ta\SE$ is $T^\cS_0 +t_s$ in the first case when the initial 
estimation errors are outside set $\cX$. The upper bound on the settling time for the second case when initial errors $(h_0,\varphi_0)\in\cX$ is straightforward to obtain directly from inequality \eqref{fts_Lyp}, and is
given by $t_s'= \frac{1}{k_0}(\cV_0)^{1-1/p}$.

The set where $\dot{\mathcal{V}} = 0$ is given by:
\begin{align}
     \dot{\mathcal{V}}^{-1}(0) &= \{(h,\varphi) \; : \; s_{K}(Q) = 0,\, y = 0,\, \Phi = 0, \, \Psi = 0\} \nonumber \\
     &= \{ (h,\varphi) \; : \; s_{K}(Q) = 0,\, Q\T \chi = 0,\, \upsilon = 0,  \, \omega = 0 \} \nonumber \\
     &= \{ (h,\varphi) \; : \; s_{K}(Q) = 0,\, \chi = 0,\, \varphi = 0 \}.
\end{align}
From Lemma \ref{costU}, $s_{K}(Q) = 0$ when $Q \in \mathcal{C}$, where $\mathcal{C}$ is the set of critical points defined in \eqref{eq:Qset}. Using the invariance-like result of theorem 8.4 in \cite{khalil}, we can 
conclude that as $t$ approaches the settling time, $(h,\varphi)$ converges to the set:
\begin{equation}
    Y = \{ (h,\varphi) \; : \; s_{K}(Q) = 0,\, \chi = 0,\, \varphi = 0 \}= \cV^{-1}(0),
\end{equation}
which is also the largest invariant set for the autonomous system of state estimation errors $(h,\varphi)\in\Ta\SE$. In the absence of measurement errors, the attitude estimation error converges to the critical set $\mathcal{C}$ of $\cU_r (Q)= \lan K, I-Q\ran$, 
which has $Q=I$ as the minimum, while the position and velocity estimation errors converge to zero.


The resulting closed-loop system of state (pose and velocity) estimation errors is H\"{o}lder-continuous with H\"{o}lder exponent $\sfrac{1}{p} < 1$. In the limiting case
of $\sfrac{1}{p} = 1$, the feedback system is Lipschitz-continuous. Based on the analysis carried out in \cite{Bohn2014observer,SanyalBohn2013}, it can be concluded that the equilibria and regions of 
attraction of the H\"{o}lder-continuous finite-time stable (FTS) observer with $p \in (1,2)$ are equivalent to those of the corresponding Lipschitz-continuous asymptotically stable observer with $p = 1$. 
As a result, the observer given by eqs. \eqref{obs1}-\eqref{obs4} is almost globally finite-time stable (AGFTS) for the state estimation errors $(h,\varphi) = (I,0)$.
\end{proof}

\textbf{Note:} An easier-to-implement version of this FTS pose estimator is given by Proposition \ref{implement-prop}, which considers the case that translational velocities are not 
measured directly. In this version, filtered point cloud and angular velocity vector measurements are used to estimate all states in $\Ta\SE$.


\section{Robustness and Pose Estimation without Translational Velocity Measurements} \label{sec: robustness}
\subsection{Robustness Analysis}
Theorem \ref{FTSestwbias} gives the convergence of almost all initial state estimation errors $(h,\varphi)$ to $(I,0)$ in a finite-time stable manner in the absence of measurement noise. In 
the presence of bounded additive measurement noise $\delta = [\delta_{\Omega},\; \delta_{\nu}]\T$ in the measurement of angular and translational velocities, the estimation errors 
will converge to a bounded neighborhood of $(h,\varphi) = (I,0)$, as the following result shows. Given the definition of $\varphi$ in \eqref{varphi}, the estimation errors in 
velocities in the presence of such measurement noise are given by:
\begin{align}
    \widetilde{\varphi}(\widehat{\mathrm{g}},\xi^m,\widehat{\xi})  =  \begin{bmatrix}
\widetilde{\omega}  \\
\widetilde{\upsilon}
 \end{bmatrix}
 &= \Ad{\widehat{\mathrm{g}}}(\xi^m - \widehat{\xi}) \nonumber \\
 &= \Ad{\widehat{\mathrm{g}}}(\xi + \delta - \widehat{\xi}) \nonumber \\
 &= \varphi(\widehat{\mathrm{g}},\xi,\widehat{\xi}) + \zeta(\widehat{\mathrm{g}},\delta),\label{varphi_noise}
\end{align}
where $\widetilde{\omega}$ and $\widetilde{\upsilon}$ are angular and translational velocity estimation errors, respectively, $\xi^m$ is given by \eqref{def_xi_m}, and
\begin{equation}
    \zeta(\widehat{\mathrm{g}},\delta) =  \begin{bmatrix}
\zeta_{\omega}  \\
\zeta_{\upsilon}
 \end{bmatrix} \in \bR^{6}
\end{equation}
represents the effect of noise in velocity measurements. The following result shows the robustness property of the pose estimator to the bounded measurement noise.
\begin{corollary}\label{corr1}
Consider the pose observer given by  \eqref{obs1}-\eqref{obs4} and measured angular and translational velocities given by \eqref{def_xi_m}. 
Let the additive noise signals in velocity measurements 
and the centroid of measured vectors in inertial frame be bounded as: 
\be 
\|\zeta_\omega\| \le \epsilon_\omega, \, 
\|\zeta_\upsilon\| \le \epsilon_\upsilon, \, \mbox{ and }\, \|\bar{q}\| \le \bar{q}_{\max}. \label{noise_barq_bnds}
\ee 
Let $\mathcal{N} \subset \mathcal{H}\times\bR^{6}\subset\SE\times\bR^6$ be the closed neighborhood of $(I,0)$ of state estimation errors defined by:
\vspace{-2mm}
\begin{multline}\label{neighborhood}
    \mathcal{N} = \{ (h,\varphi) \in \TSE: \norm{s_{L}}\leq s_{L_{\max}}, \norm{y}\leq y_{\max}, \\ 
    \norm{\Phi}\leq \Phi_{\max} \;\mbox{and}\;\norm{\Psi}\leq \Psi_{\max}\},
\end{multline}
where $\mathcal{H} = \{h \in \SE: Q \in \mathcal{S}, \chi \in \bR^3\}$ and $\cS\subset\SO$ is as defined by eq. \eqref{cSdefn} in Lemma \ref{SKQlemma}.
Select the positive observer gains such that they 
satisfy the inequality:
\vspace{-1mm}
\begin{equation}\label{robustness_ineq}
    \dfrac{\alpha_{\min}}{k_{\min}} \geq \dfrac{\Lambda}{s_{L_{\max}}^{2/p} + y_{\max}^{2/p}}, \mbox{ where } k_{\min} = \min\{k_{\upsilon}, k_{\omega}\}, 
\end{equation}
$\Lambda =\Big(\frac{2}{p}\epsilon_\omega- \Psi_{\max}\Big)\Psi^{2/p-1}_{\max} +\Big(\frac{2}{p}\big(\epsilon_\upsilon+ \bar{q}_{\max} \epsilon_{\omega}\big)-\Phi_{\max}\Big)\Phi^{2/p-1}_{\max}$,
$\alpha_{\min} = \min\{\alpha_1 k_p, \alpha_2 k_p\kappa\}$, $\norm{\zeta_{\upsilon}} \leq \epsilon_{\upsilon}$, and $\norm{\zeta_{\omega}} \leq \epsilon_{\omega}$. 
Then the estimation errors $(h,\varphi)$ converge to the bounded neighborhood $\mathcal{N}$, from almost all initial conditions except those in a set of measure zero.
\end{corollary}

\begin{proof}
    The robustness analysis for this corollary is based on the Lyapunov analysis in Theorem \ref{FTSestwbias}. 
    Consider the measured velocities and velocities estimation error in the presence of measurement noise given in \eqref{def_xi_m} and \eqref{varphi_noise}, respectively. Substituting, $\varphi$ from \eqref{varphi} into \eqref{varphi_noise}, we obtain
    \vspace{-1mm}
    \be
        \widetilde{\omega} = \omega + \zeta_{\omega}, \;\mbox{and}\; \widetilde{\upsilon} = \upsilon + \zeta_{\upsilon}. \label{noisyvelerrs}
    \ee
As the quantities $\Psi$ and $\Phi$ defined by eqs. \eqref{Psidef} and \eqref{Phidef} respectively, are linear in the velocity estimation errors, substituting eq. \eqref{noisyvelerrs} into these defining 
equations in the presence of measurement noise gives: 
\begin{align}
        \widetilde\Psi= \Psi + \zeta_\omega \mbox{ and }
        \widetilde\Phi= \Phi +\zeta_\upsilon -\bar{q}^\times\zeta_\omega, \label{tildePhiPsidefs}
\end{align}
as the modified versions of these quantities. In the presence of measurement noise, $\widetilde\omega$, $\widetilde\upsilon$, $\wt\Psi$ and $\wt\Phi$ will replace $\omega$, $\upsilon$, $\Psi$ and $\Phi$ respectively, 
in the observer equations \eqref{Psidef}-\eqref{obs2}. These substitutions are made to the Morse-Lyapunov function defined by eq. \eqref{lyapunov} and the observer equations; the modified Morse-Lyapunov function is denoted: 
\be \wt\cV=  k_p\;\mathcal{U} + \dfrac{1}{2}\wt\Psi\T \wt\Psi + \dfrac{1}{2}\wt\Phi\T \wt\Phi.
\label{modLyapf} \ee
Proceeding with the same analysis as in the first part of the proof of Theorem \ref{FTSestwbias}, the time derivative of $\wt\cV(\cdots)$ is obtained as a modification of eq. \eqref{v_dot_subs}, as follows: 
\begin{align}
    \dot{\wt{\cV}} &= - k_p \left(\alpha_{1}(s_{L}{\T} s_{L})^{1/p} +  \kappa\;\alpha_{2}(y\T y)^{1/p}\right) \nonumber \\
    &\quad - k_{\omega}\left( \wt\Psi\T \wt\Psi \right)^{1/p} - k_{\upsilon}\left( \wt\Phi\T \wt\Phi \right)^{1/p}.
    \label{dot_V_noise}
\end{align}
The difference between \eqref{dot_V_noise} and \eqref{v_dot_subs} is due to measurement noise in the velocity measurements. 
Evaluating the last two terms on the RHS of eq. \eqref{dot_V_noise}, we get:
\begin{align}
\begin{split}
    -\left( \wt\Psi\T \wt\Psi \right)^{1/p} &= -\|\Psi+ \zeta_\omega\|^{2/p} 
    \le -\Big(\|\Psi\|- \|\zeta_\omega\|\Big)^{2/p}, \\
    \mbox{and } &-\left( \wt\Phi\T \wt\Phi \right)^{1/p} = -\|\Phi +\zeta_\upsilon -\bar{q}^\times\zeta_\omega\|^{2/p} \\ \quad\quad &\le -\Big( \|\Phi\| - \|\zeta_\upsilon -\bar{q}^\times\zeta_\omega\|\Big)^{2/p}
\end{split} \label{mod-PsiPhi-ineqs}    
\end{align}
Now, considering the first inequality in \eqref{mod-PsiPhi-ineqs} and applying the binomial expansion to the RHS, we obtain:
\begin{align}
     &-\left( \wt\Psi\T \wt\Psi \right)^{1/p} \le 
    -\|\Psi\|^{2/p} \left( 1-\norm{\frac{\zeta_\omega}{\Psi}} \right)^{2/p} \nn \\
    &= -\|\Psi\|^{2/p} \left( 1- \frac{2}{p}\norm{\frac{\zeta_\omega}{\Psi}}+ \frac{\frac{2}{p}\big(\frac{2}{p}-1\big)}{2!}\norm{\frac{\zeta_\omega}{\Psi}}^2 +\ldots \right) \label{modPsi-ineq} \\
    &\le -\|\Psi\|^{2/p} \left( 1- \frac{2}{p}\norm{\frac{\zeta_\omega}{\Psi}} \right) \mbox{ as } \frac{2}{p}\in ]1,2[ \mbox{ if } p\in ]1,2[. \nn
\end{align}
The series expansion in the second line of \eqref{modPsi-ineq} converges if $\norm{\frac{\zeta_\omega}{\Psi}}<1$, which would be the case initially if the initial state estimation errors are larger in magnitude than 
the angular velocity measurement noise. But that is not required for the inequality in the last line of \eqref{modPsi-ineq} to hold, because all terms in parentheses in the second line from the 
third term onward, are positive. A similar simplified inequality can be obtained from the second of inequalities \eqref{mod-PsiPhi-ineqs}, as follows: 
\begin{align}
    -\left( \wt\Phi\T \wt\Phi \right)^{1/p} &\le -\|\Phi\|^{2/p} \left( 1- \frac{2}{p}\norm{\frac{\zeta_\upsilon - \bar{q}^\times\zeta_\omega}{\Phi}} \right) \nn \\ &\le -\|\Phi\|^{2/p} \left( 1- \frac{2}{p} \frac{\|\zeta_\upsilon\| + \| \bar{q}^\times\zeta_\omega\|}{\|\Phi\|} \right).  
    \label{modPhi-ineq}
\end{align}
Substituting inequalities \eqref{modPsi-ineq}-\eqref{modPhi-ineq} into the expression \eqref{dot_V_noise} for the time derivative of the modified Morse-Lyapunov function, we get:
\begin{align}
    \dot{\wt{\cV}} \le& -\alpha_1 k_p \|s_L\|^{2/p} -\alpha_2 k_p\kappa \|y\|^{2/p} 
    -k_\omega \|\Psi\|^{2/p} \left( 1- \right. \nn \\ 
    &\left. \frac{2}{p}\norm{\frac{\zeta_\omega}{\Psi}} \right) -k_\upsilon \|\Phi\|^{2/p} \left( 1- \frac{2}{p} \frac{\|\zeta_\upsilon\| + \| \bar{q}^\times\zeta_\omega\|}{\|\Phi\|} \right). \label{Vdot_noise}
\end{align}


From \eqref{Vdot_noise}, the upper bound on $\dot{\mathcal{V}}$ along the boundary of the neighborhood $\mathcal{N}\ni (I,0)$ as defined in \eqref{neighborhood} with noise as bounded by \eqref{noise_barq_bnds}, is given by
\begin{align*}
\begin{split}
    \dot{\wt{\mathcal{V}}} \leq & -\alpha_{\min}\left( s_{L_{\max}}^{2/p} + y^{2/p}_{\max} \right) -k_{\min} \Psi^{2/p-1}_{\max}\Big(\Psi_{\max} -\frac{2}{p}\epsilon_\omega\Big) \nn \\ & -k_{\min} \Phi^{2/p-1}_{\max}\Big(\Phi_{\max} -\frac{2}{p}\big(\epsilon_\upsilon+ \bar{q}_{\max} \epsilon_{\omega}\big)\Big). 
    \end{split}
\end{align*}
where $\alpha_{\min} = \min\{\alpha_1 k_p, \alpha_2 k_p\kappa\}$. 
A sufficient condition for the convergence of all trajectories starting outside of the neighborhood $\mathcal{N}$ of $(I,0)$ to this neighborhood, is given by the RHS of this inequality being upper bounded 
by zero. This leads to the expression
\begin{align}
    \begin{split}\label{condition1}
        &\dfrac{\alpha_{\min}}{k_{\min}} \geq \dfrac{\Psi^{2/p-1}_{\max}\rho_\omega + \Phi^{2/p-1}_{\max}\rho_\upsilon}{s_{L_{\max}}^{2/p} + y_{\max}^{2/p}}, \\
        &\mbox{where } \rho_\omega= \frac{2}{p}\epsilon_\omega- \Psi_{\max} \mbox{ and } \\ &\rho_\upsilon= \frac{2}{p}\big(\epsilon_\upsilon+ \bar{q}_{\max} \epsilon_{\omega}\big)-\Phi_{\max}. 
    \end{split}    
\end{align}
Inequality \eqref{condition1} leads to \eqref{robustness_ineq}, and relates the ratio $\sfrac{\alpha_{\min}}{k_{\min}}$ and known bounds on measurement noises  $(\zeta_\upsilon,\zeta_\omega)$ to the desired 
neighborhood for convergence of state estimation errors $\mathcal{N}$. 
Its satisfaction guarantees the convergence of state estimation errors $(h,\varphi)\in\Ta\SE$ from almost all initial conditions to the neighborhood $\mathcal{N}$ of $(I,0)$ given by \eqref{neighborhood}.
\end{proof}

\begin{remark}
    Bounds on state estimation errors $(h,\varphi)$ lead to the bounds $s_{L_{\max}}, y_{\max}, \Phi_{\max} \mbox{ and } \Psi_{\max}$, while  $\bar{q}_{\max}$ is a bound on the sensor's position measurement. 
\end{remark}

\subsection{FTS Pose Estimator Implementation without Translational Velocity Measurements}
The FTS pose estimator is implemented for the case when only angular velocity measurements are available along with position measurements of the feature points. This is a common case of practical interest because rate gyros in inertial measurement units can directly measure angular velocity but translational velocity is not 
directly measured by sensors onboard aerial vehicles. To implement the FTS pose estimator in this case, a finite-time stable filter is used to 
obtain filtered values of the measured feature points and measured angular velocity. The filtered feature points are used to construct a filtered 
estimate of the translational velocity. The filtered versions of the translational and angular velocity vectors are used in this implementation 
of the FTS pose estimator, as given in the following proposition. \\
\begin{proposition}\label{implement-prop}
    Consider the case where position vectors of feature points in 3D point clouds and the angular velocity vector are measured by onboard sensors, but translational velocity is not measured. 
    Let $z^{f}$ denote the filtered value of the measured vector $z^{m}$ at time $t$ ($z^m$ could stand for $\Omega^m$ or $a_i^m$). 
Construct the filtered translational velocity and filtered mean of the measured feature points, as follows:
\begin{align} 
&\nu^{f} = \dfrac{1}{j} \sum^{j}_{i=1} \big((a^{f}_{i})^{\times}\Omega^{f} -  \mathit{v}^{f}_{i}\big), \mbox{ where } v_i^f=\dot a_i^f, \label{nufdef} \\
&\bar{a}^f= \dfrac{1}{j} \sum^{j}_{i=1} a_i^f, \dot{\bar{a}}^f= \dfrac{1}{j} \sum^{j}_{i=1}\dot a_i^f, \mbox{ and } \xi^f= \bbm \Omega^f \\ \nu^f\ebm. \label{barafdef}
\end{align} 
Thereafter, replace the measured quantities $\bar{a}^{m}$ 
and $\xi^m$ in equations \eqref{z2def}, \eqref{Phidef}, and \eqref{obs3} in Theorem \ref{FTSestwbias} by the filtered quantities 
$\bar{a}^{f}$ 
and $\xi^f$, respectively. This gives the implementation of the FTS pose estimator when translational velocity is not directly measured.
\end{proposition}
\begin{proof}
Consider the case where translational velocity measurements are unavailable. Rigid body velocities can be computed using inertial and point cloud measurements. To do this, differentiate equation \eqref{position_vec} to obtain the following:
\begin{align}
    \dot{q}_{i} = \dot{R}a_{i} + R\dot{a}_{i} + \dot{b} = R(\Omega^{\times}a_{i} + \dot{a}_{i} + \nu) = 0,
\end{align}
because the $q_i$ are inertially fixed vectors. This leads to
\vspace{-1mm}
\begin{align}\label{vi}
    \mathit{v}_{i} = \dot{a}_{i} = [a^{\times}_{i} \; -I]\xi,
\end{align}
where $[a^{\times}_{i} \; -I]$ has full row rank and $\xi$ is defined in \eqref{def_g}. If angular velocity and point cloud measurements are available, the translational velocity of the rigid body can be obtained by filtering the measured points and their velocities and rewriting \eqref{vi} as
\begin{equation}\label{v_f}
    \nu^{f} = \dfrac{1}{j} \sum^{j}_{i=1} \big((a^{f}_{i})^{\times}\Omega^{f} -  \mathit{v}^{f}_{i}\big),
\end{equation}
averaging over the $i$ measured points.
Therefore, the rigid body's filtered velocities $\xi^{f}$ can be expressed as 
\begin{equation}
    \xi^{f} = \begin{bmatrix}
\Omega^{f}  \\
\dfrac{1}{j} \sum^{j}_{i=1} \big((a^{f}_{i})^{\times}\Omega^{f} -  \mathit{v}^{f}_{i}\big)
 \end{bmatrix}.
\end{equation}
This concludes the proof.
\end{proof}
A finite-time stable filtering scheme for the vector measurements $z^m$ in discrete time, is given by Proposition \ref{filterprop} in section \ref{sec:sim}.

\section{Selected Pose Estimation Schemes for Comparison}\label{sec:comparison}
The proposed FTS pose estimation scheme is compared with two other pose estimation schemes: the variational pose estimator (VPE) and dual quaternion multiplicative extended Kalman filter (DQ-MEKF).

\subsection{Variational Pose Estimatior}
This variational pose estimator in \cite{izadi2016rigid} is obtained by applying the Lagrange–d’Alembert principle from variational mechanics to the Lagrangian obtained from measurement residuals along with a Rayleigh dissipation term linear in the velocity
measurement residuals. The discretized variational pose estimator is given by:
\begin{align}
    &(J\omega_{k})^{\times} = \dfrac{1}{\Delta t} \big( F_{k}\mathcal{J} - \mathcal{J}F\T_{k} \big),\ k\in\bN, \\
    &(M + \Delta t \mathbb{D}_{t})\upsilon_{k+1} = F\T_{k}M\upsilon_{k} \nonumber \\ 
    &\qquad + \Delta t\; \kappa (\hat{b}_{k+1} + \hat{R}_{k+1}\bar{a}^{m}_{k+1} - \bar{q}_{k+1}), \\
    &(J + \Delta t \mathbb{D}_{r})\omega_{k+1} = F\T_{k}J\omega_{k} + \Delta t M\upsilon_{k+1}\times \upsilon_{k+1} \nn \\
    &\qquad + \Delta t\; \kappa \;\bar{q}^{\times}_{k+1} (\hat{b}_{k+1} + \hat{R}_{k+1}\bar{a}^{m}_{k+1})\nn \\
    &\qquad - \Delta t\Phi' \Big( \mathcal{U}^{0}_{r} (\hat{\mathrm{g}}_{k+1}, E^{m}_{k+1}, D_{k+1}) \Big) S_{\Gamma_{k+1}}(\hat{R}_{k+1}), \\
    &\hat{\xi}_{k} = \xi^{m}_{k} - \mbox{Ad}_{{\hat{\mathrm{g}}}^{-1}_{k}}\varphi_{k}, \\
    &\hat{\mathrm{g}}_{k+1} = \hat{\mathrm{g}}_{k} \: \mbox{exp}(\Delta t\hat{\xi}^{\vee}_{k}) \label{gkupdate},
\end{align}
where $F_{k} \in \SO$, $\Delta t$ is the time step, $J, M, \mathbb{D}_{t}, \mathbb{D}_{r} \in \bR^{3\times 3}$  are positive definite matrices, $\mathcal{J} = \dfrac{1}{2}\mbox{trace}[J]I - J$, 
$\left(\hat{\mathrm{g}}(t_{0}), \hat{\mathrm{\xi}}(t_{0})\right) = (\hat{\mathrm{g}}_{0}, \hat{\mathrm{\xi}}_{0})$, $\varphi_{k} = [\omega_{k}\T \; \upsilon_{k}\T]\T$. The variables 
$\xi^{m}_{k}, \bar{a}^{m}_{k}, \bar{q}_{k}, E^{m}_{k}, D_{k}$ are the discrete-time values of the quantities defined in Section \ref{sec:Problem}. Like the continuous time FTS-PE of Theorem \ref{FTSestwbias}, it maintains the geometry of the state space of rigid body 
motions. This is because the exponential map $\exp: \se\to\SE$ is used to update the pose estimate according to eq. \eqref{gkupdate}. This exponential map is numerically evaluated using 
Rodrigues' rotation formula, as given in~\cite{bumu95, bullo2019geometric}.
 
\subsection{DQ-MEKF estimator}
The dual quaternion multiplicative extended Kalman filter (DQ-MEKF) is proposed in \cite{filipe2015extended}. This filter in the absence of bias is given by:
\begin{align}
    &\frac{\di}{\di t} (\hat{\boldsymbol{q}}_{B/I}) \approx \hat{\boldsymbol{q}}^{*}_{B/I}\boldsymbol{q}_{B/I}, \\
    &\hat{\boldsymbol{\omega}}^{B}_{B/I} = \boldsymbol{\omega}^{B}_{B/I,m},
\end{align}
where $\boldsymbol{q}_{B/I}$ is the unit dual quaternion of a body frame with respect to an inertial frame and is represented by 
\begin{align}
    \boldsymbol{q}_{B/I} = q_{B/I,r} + \epsilon q_{B/I,d},
\end{align}
and ${\boldsymbol{\omega}}^{B}_{B/I}$ is the dual velocity of the body frame with respect to the inertial frame expressed in the body frame. The following constraints are enforced to avoid numerical errors in the propagation of $\wh{\boldsymbol{q}}_{B/I}$:
\begin{align}
    [q_{B/I,r}] &= \dfrac{[q_{B/I,r}]}{\norm{[q_{B/I,r}]}} \quad\quad \mbox{and} \; \nn \\
    [q_{B/I,d}] &= \Biggl(I_{4\times 4} - \dfrac{[q_{B/I,r}][q_{B/I,r}]\T}{\norm{[q_{B/I,r}]}^{2}}\Biggl)[q_{B/I,d}].
\end{align}
The first equation above is a normalization of the rotation part of the dual quaternion that forces it to be a unit quaternion, and the second is a projection of $[q_{B/I,d}]$ on the 
orthogonal subspace of $[q_{B/I,r}]$. Together, these steps ensure that the dual quaternion maps to a pose on $\SE$, and they add to the numerical costs of this DQ-MEKF scheme. 
Thereafter, the covariance matrix $P_{6\times 6}$ of the state satisfies the Riccati equation and is propagated as follows 
\begin{align}
    &\dot{P}_{6\times 6}(t) = F_{6\times 6}(t)P_{6\times 6}(t) + P_{6\times 6}(t)F\T_{6\times 6}(t), \nonumber \\
    &\qquad\qquad G_{6\times 6}(t)Q_{6\times 6}(t)G\T_{6\times 6}(t).
\end{align}
The state estimate at time $t_{k}$ is then calculated as
\begin{align}
    &\hat{\boldsymbol{q}}^{+}_{B/I}(t_{k}) = \hat{\boldsymbol{q}}^{-}_{B/I}(t_{k}) \Delta^{\star} \delta\hat{\boldsymbol{q}}_{B/I}(t_{k}),
\end{align}
where
\begin{align}
    &\boldsymbol{q}_{B/I} = q_{B/I,r} + \epsilon q_{B/I,d}, \\
    &\Delta^{\star} \delta\hat{\boldsymbol{q}}_{B/I} = \biggl( \sqrt{1 - \norm{\Delta^{\star} \overline{\delta \hat{q}_{B/I,r}}}^{2}}, \Delta^{\star} \overline{\delta \hat{q}_{B/I,r}}  \biggl) \nonumber \\
    &\quad + \epsilon \Biggl(\dfrac{-\Delta^{\star} \overline{\delta \hat{q}_{B/I,r}}\T \Delta^{\star} \overline{\delta \hat{q}_{B/I,d}}}{\sqrt{1 - \norm{\Delta^{\star} \overline{\delta \hat{q}_{B/I,r}}}^{2}}},  \Delta^{\star} \overline{\delta \hat{q}_{B/I,d}}   \Biggl),
    \label{update1}
\end{align}
when the attitude error between the true and estimated attitude is less than 180 deg. If the attitude error is greater than 180 deg then,
\begin{align}
    &\Delta^{\star} \delta\hat{\boldsymbol{q}}_{B/I} = \Biggl( \dfrac{1}{\sqrt{1 + \norm{\Delta^{\star} \overline{\delta \hat{q}_{B/I,r}}}^{2}}} , 
\dfrac{\Delta^{\star} \overline{\delta \hat{q}_{B/I,r}}}{\sqrt{1 + \norm{\Delta^{\star} \overline{\delta \hat{q}_{B/I,r}}}^{2}}}  \Biggl) \nonumber \\
    &\quad + \epsilon \Biggl(\dfrac{-\Delta^{\star} \overline{\delta \hat{q}_{B/I,r}}\T \Delta^{\star} \overline{\delta \hat{q}_{B/I,d}}}{1 / \sqrt{1 + \norm{\Delta^{\star} \overline{\delta \hat{q}_{B/I,r}}}^{2}}},  \Delta^{\star} \overline{\delta \hat{q}_{B/I,d}}   \Biggl). \label{update2}
\end{align}
Unlike the VPE and FTS-PE, the DQ-MEKF does not maintain the geometry of the state space $\Ta\SE$ and it avoids unwinding, as given in~\cite{bhat2000topological}, by using the discontinuous update law given by 
equations \eqref{update1}-\eqref{update2}.

\section{Results from Numerical Simulations and Experiment}\label{sec:sim}
This section provides numerical simulation results and results from an experiment, for the proposed FTS-PE scheme and compares its numerical performance to the observers in section \ref{sec:comparison}. These results  
demonstrate its performance and the convergence of estimation errors. The proposed FTS pose estimation scheme has been discretized to enable onboard computer
implementation, and the discretization is based on a geometric integration scheme. Geometric variational integration schemes, in contrast to commonly used numerical integration methods 
such as (unstructured) Runge-Kutta, preserve the geometry of the state space, which in this case is the tangent bundle $\Ta\SE$, without requiring any projection or parametrization. 
Furthermore, they maintain energy-momentum properties for (corresponding continuous-time) conservative and dissipative systems. Let $\dot{\omega} = \gamma$ and $\dot{\upsilon} = \eta$, where $\gamma, \eta$ are the right sides 
of \eqref{obs1}, \eqref{obs2} respectively. Let $\Delta t=t_{k+1}-t_k$ be the time step size and subscript $k$ denote a quantity evaluated at time $t_k$. Then, the discretized observer equations are obtained as:
\begin{align}
    &\omega_{k+1} = \omega_{k} + \Delta t\gamma_{k}, \label{dis_obs_1} \\
    &\upsilon_{k+1} = \upsilon_{k} + \Delta t\eta_{k}, \label{dis_obs_2} \\
    &\wh{\xi}_{k} = \xi^{m}_{k} -\Ad{{\wh{\mathrm{g}}}^{-1}_{k}}\varphi_{k}, \label{dis_obs_3} \\
         &{\wh g}_{k+1} = \wh g_{k}\mbox{exp}(\Delta t\wh\xi^{\vee}_{k}), \label{dis_obs_4}
\end{align}
where
\begin{align}\label{dis_gamma}
    \gamma_{k} &= - k_p s_{L_{k}} - k_{\omega}\frac{\Psi_{k}}{\big({\Psi_{k}}\T\Psi_{k}\big)^{1-1/p}} \nn \\
    &-\frac{\alpha_{1} }{\big({s_{L_{k}}}\T s_{L_{k}}\big)^{1-1/p}}H(s_{L_{k}}) w_{L_{k}}, 
\end{align}
\begin{align}\label{dis_eta}
    \eta_{k} &=  \bar{q}^\times \gamma_{k} - k_{p}\; \kappa \;y_{k} - k_{\upsilon}\frac{\Phi_{k}}{\big({\Phi_{k}}\T\Phi_{k}\big)^{1-1/p}}\nn \\
    &- \frac{\alpha_{2} }{\big({y_{k}}\T y_{k}\big)^{1-1/p}}H(y_{k}) v_{y_{k}}, 
\end{align}
and \be \varphi_{k} = \bbm \omega_{k} \\ \upsilon_{k}\ebm. \ee
\begin{proposition}\label{filterprop}
    In the absence of translational velocity measurements, 
    a finite-time stable filter 
    in discrete time is obtained according to \cite{Sanyal2022}, as follows:
    \begin{align}
        &z^{f}_{k+1} = z^{m}_{k} + \cD(c_{k})c_{k}
        + (\cD(z^{f}_{k} - z^{f}_{k-1}))(z^{f}_{k} - z^{f}_{k-1}), \\
        &\dot{z}^{f}_{k} = \dfrac{(z^{f}_{k+1} - z^{f}_{k})}{\Delta t},\\
        &\mbox{where} \; \cD(c_{k}) = \dfrac{(c_{k}\T c_{k})^{1-1/r} - \lambda_{c}}{(c_{k}\T c_{k})^{1-1/r} + \lambda_{c}} \mbox{and} \; c_k=z^f_k- z^m_k,
    \end{align}
    $r \in ]1,2[$ and $\lambda_{c} > 0$ are constants, $z^{f}_{k} = z^{f}(t_k)$ and $z^{m}_{k} = z^{m}(t_k)$, respectively. Thereafter, $z^{m}_{k}$ is replaced with $z^{f}_{k}$ in eq. \eqref{dis_obs_3} of the 
    discrete-time observer.
\end{proposition}
The matrix exponential map in \eqref{dis_obs_4} ensures that the pose at each instant is in $\SE$ provided the initial estimate is in $\SE$.
The proposed estimation scheme is simulated with a time step size of $\Delta t = 0.1\mbox{s}$ for a duration of $T = 30s$. For numerical simulations, the initial attitude and position of the rigid body are selected as:
\begin{equation}
    R_{0} = I,\; b_{0} = [0 \; 0 \; 0]\T \mbox{m}.
\end{equation}
The initial angular and translational velocities of the rigid body, respectively, are selected as:
\begin{equation}
    \Omega_{0} = [0 \; 0.15 \; 0]\T \;\;\mbox{rad/s} \;\;\mbox{and}\; \nu_{0} = [0.65 \; 0 \; 0.1]\T \;\mbox{m/s}.
\end{equation}
We assume that point cloud measurements are obtained from a vision sensor for points whose inertial positions are known. In addition, angular velocities measurements are obtained from rate gyros. The filtered position measurements and 
velocities are then obtained using the FTS filter given by Propositions \ref{implement-prop} and \ref{filterprop}.

The initial state estimates are selected to be:
\begin{align}\label{Rhat0}
    &\wh{R}_{0} = \mbox{exp}_{\SO} \left(0.9\pi \left([1\;0\;0]\T \right)^{\times}\right), \;\; \wh{b}_{0} = [1.5\;1\;1]\T, \;\; \nonumber \\
    &\wh{\Omega}_{0} = [-0.67 \; -0.25 \; -0.09]\T \;\;\mbox{rad/s}
    \;\;\mbox{and}\; \nonumber \\ &\wh{\nu}_{0} = [0.76 \; -2.63 \; 2.83]\T \;\mbox{m/s}.
\end{align}
The observer gains selected are $k_p = 10.1, \; k_{\upsilon} = 10.02, \; k_{\omega} = 11.01, \; p = 13/11, \; \kappa = 1.1, \; \alpha_{1} = 88.65, \; \alpha_{2} = 0.9609$. The proposed FTS pose estimation scheme is simulated for the following two different cases.
\subsection{Simulation results of FTS pose estimator in the absence and presence of measurement noise in angular and translational velocities}\label{8_1}
This subsection presents the simulation results of the FTS-PE scheme in the absence and presence of measurement noise in angular and translational velocities, with initial estimation errors. First, 
the performance of the FTS-PE scheme is evaluated in the absence of measurement noise. The true and estimated trajectories over the time interval are depicted in Fig. \ref{fig:3d_traj}. The attitude estimation error is parameterized by the principal 
rotation angle $\phi$ of the attitude estimation error matrix $Q$ as $\phi = \textup{cos}^{-1} (\frac{1}{2}(\textup{tr}(Q)-1))$. The principal angle $\phi$ and position estimation error $\chi$ are displayed in Fig. \ref{fig:Q_chi}, which shows 
the finite-time convergence of the estimation errors for both attitude and position. 


\begin{figure}[h]
    \centering
    \includegraphics[width=0.95\columnwidth]{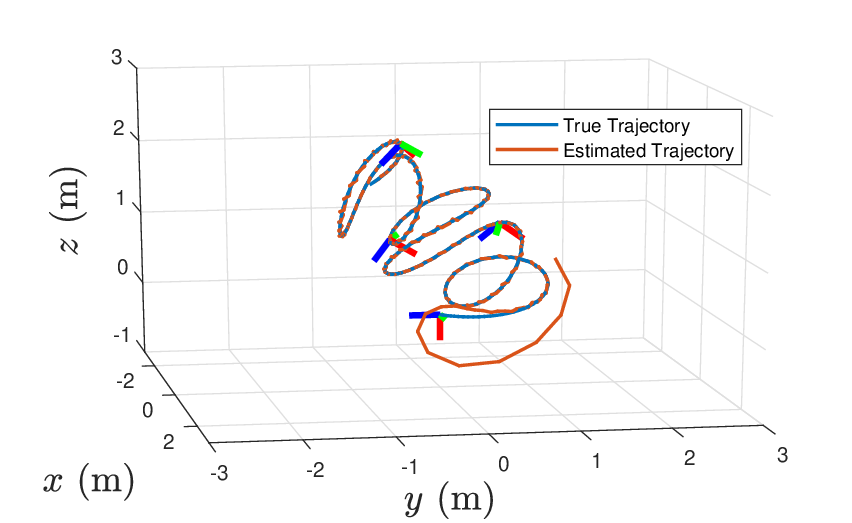}
    \caption{True and Estimated Trajectories.}
    \label{fig:3d_traj}
\end{figure}

\begin{figure}[h]
    \includegraphics[width=\columnwidth]{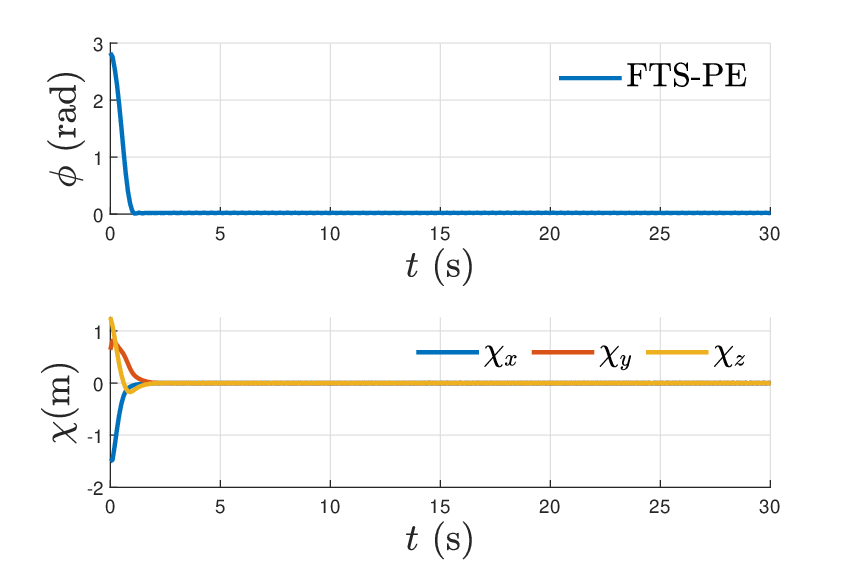}
    \caption{Plots of attitude and position estimation errors in the absence of measurement noise.}
    \label{fig:Q_chi}
\end{figure}
%
Next, the effects of measurement noise in angular and translational velocities on the performance of the FTS-PE scheme are considered. We assume that both angular and translational velocity 
measurements are corrupted by Gaussian noise with zero mean and standard deviation of $9.1673\;^{o}\mbox{/s}$ and $0.02\;\mbox{m/s}$, respectively. 
The angular and translational velocity estimation errors in the presence of measurement noise in velocities are shown in Fig. \ref{fig:varphi_rv}, while Fig. \ref{fig:Q_chi_rv} displays the principal angle of the attitude estimation error 
$\phi$ and the position estimation error $\chi$. The simulation results depicted in these figures demonstrate the stability of the proposed FTS-PE scheme and the convergence of estimation errors to 
a small, bounded neighborhood of $(h, \varphi) = (I, 0)$ despite persistent measurement noise in velocities at these relatively high noise levels. The size of this neighborhood depends on the bounds 
of the measurement noise and the observer gains selected, according to Corollary \ref{corr1}. In Fig. \ref{fig:zeta_rv}, the plots show the bounded values of $\zeta_{\omega}$ and $\zeta_{\upsilon}$, 
along with the bounds determined by inequality \eqref{robustness_ineq} in Corollary \ref{corr1}.
%
%
\begin{figure}[h]
    \includegraphics[width=\columnwidth]{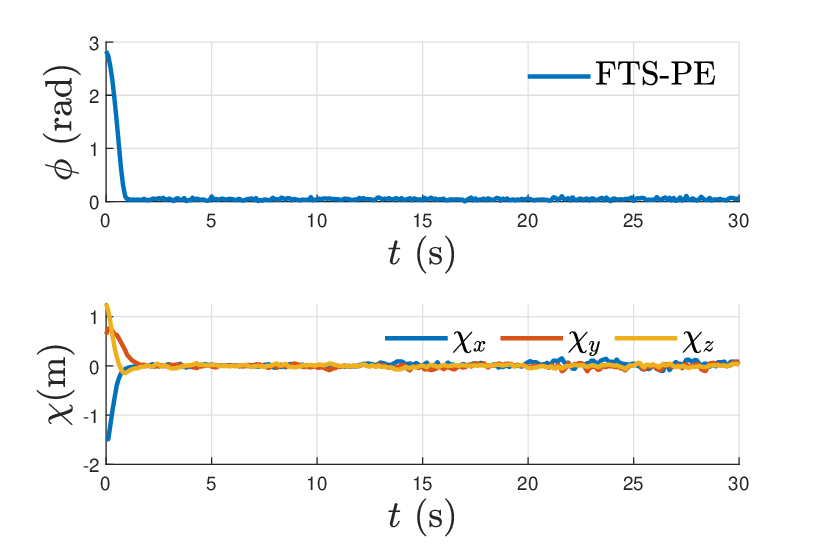}
    \caption{Plots of attitude and position estimation errors in the presence of measurement noise in velocities.} 
    \label{fig:Q_chi_rv}
\end{figure}
\begin{figure}[h]
    \includegraphics[width=\columnwidth]{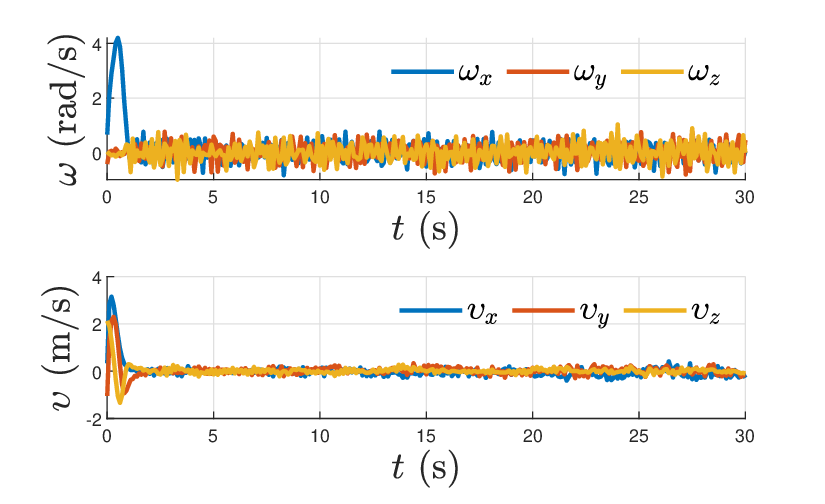}
    \caption{Plots of estimation errors in velocities in the presence of measurement noise in velocities.}   
    \label{fig:varphi_rv}
\end{figure}
\begin{figure}[ht]
 \includegraphics[width=\columnwidth]{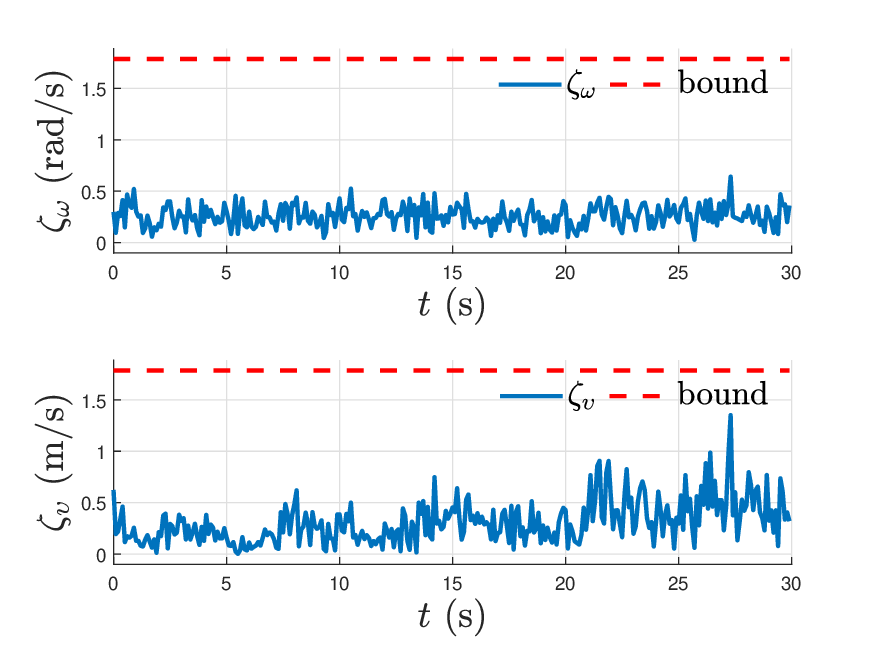}
    \caption{Plots of bounded measurement noise in velocities.}   
    \label{fig:zeta_rv}
\end{figure}

\subsection{Comparison of FTS-PE with other pose estimators in the presence of measurement noise in all states}
Here, the effects of vector measurement errors are considered in addition to measurement noise in angular and translational velocities. The performance of the FTS-PE 
is compared to the two pose estimators outlined in section \ref{sec:comparison}, by utilizing the same initial estimated states used in subsection \ref{8_1}. The time step size for this simulation 
is $\Delta t = 0.1\mbox{s}$. The vision (point cloud) measurements are obtained at this constant rate by body-fixed sensors with an additive uniform random noise with zero mean and 
standard deviation of $0.15\mbox{m}$. The noise in angular and translational velocity measurements is similar to that considered in subsection \ref{8_1}. 
For a fair comparison, the same measurements are used for all three pose estimation schemes. The tuned gain parameters for both DQ-MEKF and VPE are unchanged and are as given in \cite{filipe2015extended} and \cite{izadi2016rigid}, respectively.
Fig. \ref{fig:Q_chi_noise} plots the position estimation error (bottom) and the principal angle of the attitude estimation error (top). The results from the simulations presented in these 
plots support the guaranteed stable and robust performance of FTS-PE and VPE in the presence of noisy measurements, unlike DQMEKF. Moreover, FTS-PE shows convergence of estimation errors to the 
smallest neighborhood of $(h, \varphi) = (I, 0)$ in finite time.

The simulation run-times for all three pose estimators corresponding to the results in Fig. \ref{fig:Q_chi_noise}, are given in Table \ref{sim_time}. From this table, one can see that 
the FTS-PE executes the fastest of all three pose estimators compared. 
In addition, the root-mean-square (RMS) values of attitude and position estimation errors over the simulated duration of 30s for these three pose estimators are given in Table \ref{avg_error}. 
These RMS values show that the pose estimation errors for the FTS-PE are the lowest. These pose estimators were implemented on a computer with a 3.30 GHz, AMD Ryzen 9 5900HS CPU, and 16 GB of RAM.

\begin{table}[htb]
\centering
\begin{tabular}{|c|c|c|c|} 
 \hline
 Estimator & FTS-PE & DQMEKF & VPE\\ 
 \hline
 Run-times & 0.0469 s & 0.1250 s & 0.1094 s\\ 
 \hline
\end{tabular}
\caption{Simulation times for FTS-PE, DQMEKF, and VPE.}
\label{sim_time}
\end{table}
\begin{table}[htb]
\centering
\begin{tabular}{|c|c|c|c|} 
 \hline
 Estimator & FTS-PE & DQMEKF & VPE\\ 
 \hline
 Attitude error (rad) & 0.3544 & 0.6806 & 0.5780\\ 
 \hline
 Position error (m) & 0.2769 & 0.4336 & 0.4102\\ 
 \hline
\end{tabular}
 \vspace{2mm}
\caption{RMS values of attitude and position estimation errors for FTS-PE, DQMEKF, and VPE.}
\label{avg_error}
\end{table}
%
%

\begin{figure}[h]
    \includegraphics[width=\columnwidth]{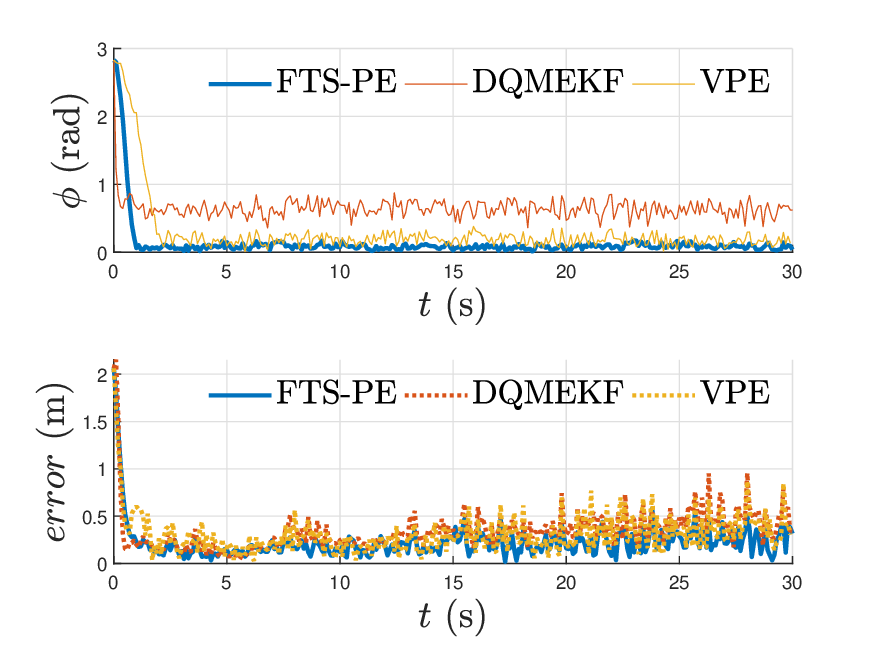}
    \caption{Plots of attitude and position estimation errors in the presence of measurement noise in all states.} 
    \label{fig:Q_chi_noise}
\end{figure}
\subsection{Experimental results for FTS pose estimator using ZED 2i sensor}
This subsection presents results obtained from an experimental evaluation of the (discretized) finite-time stable pose estimator (FTS-PE) given by eqs. \eqref{dis_obs_3}-\eqref{dis_eta}, using a depth camera sensor. 
For the experiment, the ZED 2i sensor, an advanced commercially available stereo depth camera, was used. The FTS-PE provides estimates of the pose and velocities directly 
from three-dimensional point cloud measurements and angular velocity measurements obtained by this sensor. The vector measurements for FTS-PE are sets of points in the 3D point clouds captured by the 
stereo camera. The 3D spatial data, in the form of the captured point clouds, are stored in ROS bag files and accessed by extracting point cloud messages from the relevant ROS topic. Furthermore, 
the ROS topic "$'/zed2i/zed\_node/imu/data$" publishes messages of type “$sensor\_msgs/Imu$”, which contain fused inertial vector data generated from the ZED 2i's inbuilt Inertial Measurement Unit 
(IMU). The angular velocity of the rigid body is extracted from this IMU data. 

As angular velocity and point cloud measurements are available, the translational velocity is obtained by filtering the measured points and their velocities using eq. \eqref{v_f}. Furthermore, the filtered mean of 
the measured feature points is obtained from eq. \eqref{barafdef} in Proposition \ref{implement-prop}, and the discretized expressions for the continuous-time FTS filter given by Proposition \ref{filterprop} are used for this implementation.
The FTS-PE is implemented with a time step size of $\Delta t = 0.0702$ s for a duration of $T = 30.82$ s, and $\lambda_{c} = 1$. This corresponds to a point cloud sampling frequency of 14.25 Hz for the ZED 2i stereo camera.


For this experiment, the initial attitude and position estimates were set equal to the values given in \eqref{Rhat0}. The estimated values of initial angular and translational velocities were respectively
\begin{align}\label{Vhat0}
    &\wh{\Omega}_{0} = [ -0.7500\;  -0.4910\; -0.1070]\T \;\;\mbox{rad/s}
    \;\;\mbox{and}\; \nonumber \\ &\wh{\nu}_{0} = [ -0.9501\; -1.2812 \;3.0536 ]\T \;\mbox{m/s}.
\end{align}
The first 3D point cloud image obtained by the sensor was considered to define the inertial reference frame. Pose and velocities were estimated from subsequent point cloud image frames with reference to the first frame as 
the inertial frame. This is also the case for  pose estimates provided by the software (SDK) of the sensor.


Fig  \ref{fig:pose estimation errors experimental}. 
displays the principal rotation angle of the attitude estimation error and the norm of the position vector estimation error.
The angular and translational velocity estimation errors obtained from this
experimental evaluation of the FTS-PE scheme are shown in Fig.  \ref{fig:velocity estimation errors experimental}.
\begin{figure}[h]
     \centering
   \includegraphics[width=\columnwidth]{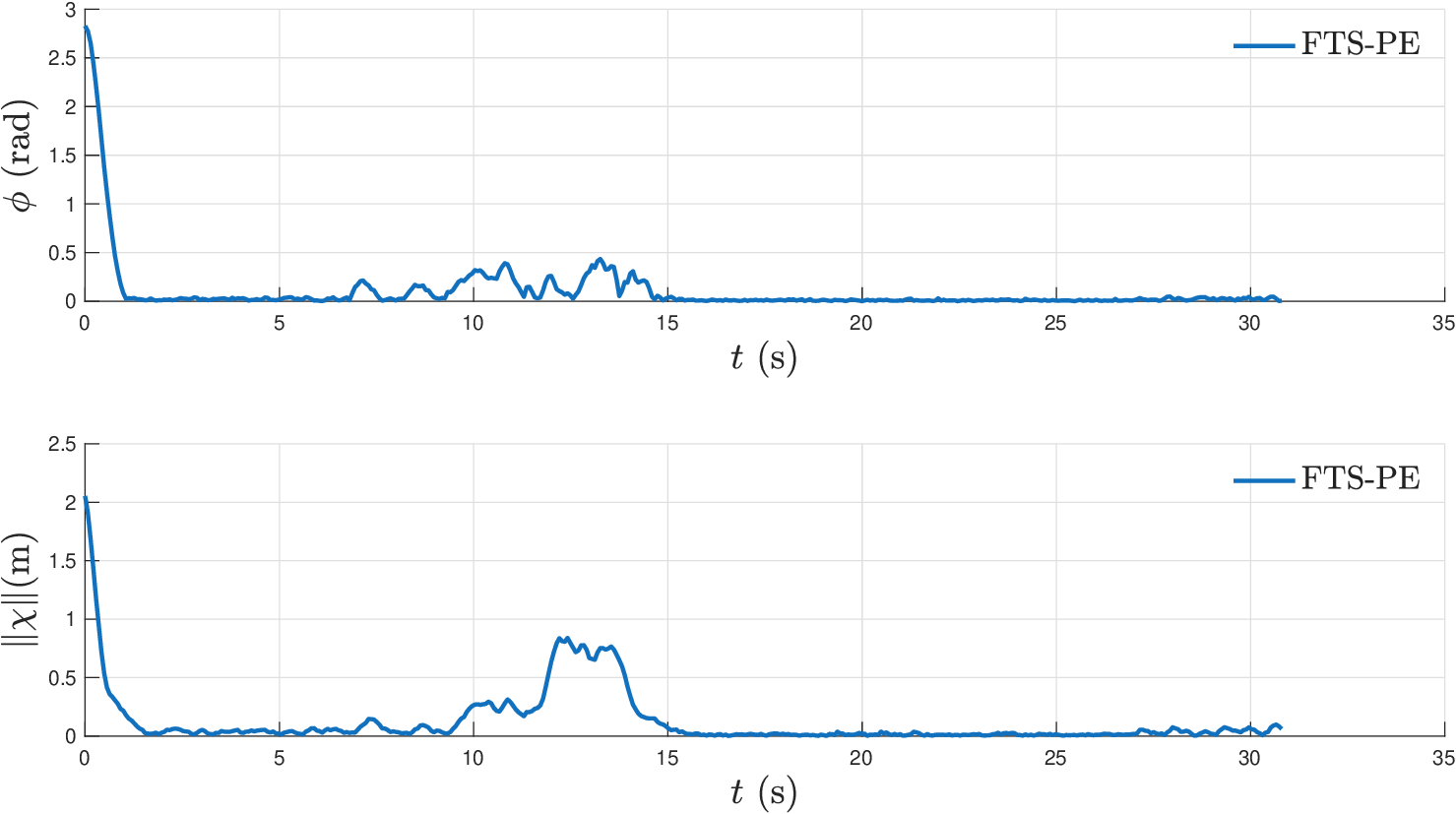}
    \caption{Plots of attitude estimation error and norm of position vector estimation error obtained from the experimental evaluation of the FTS-PE scheme.} 
    \label{fig:pose estimation errors experimental}
\end{figure}
\begin{figure}[h]
   \includegraphics[width=\columnwidth]{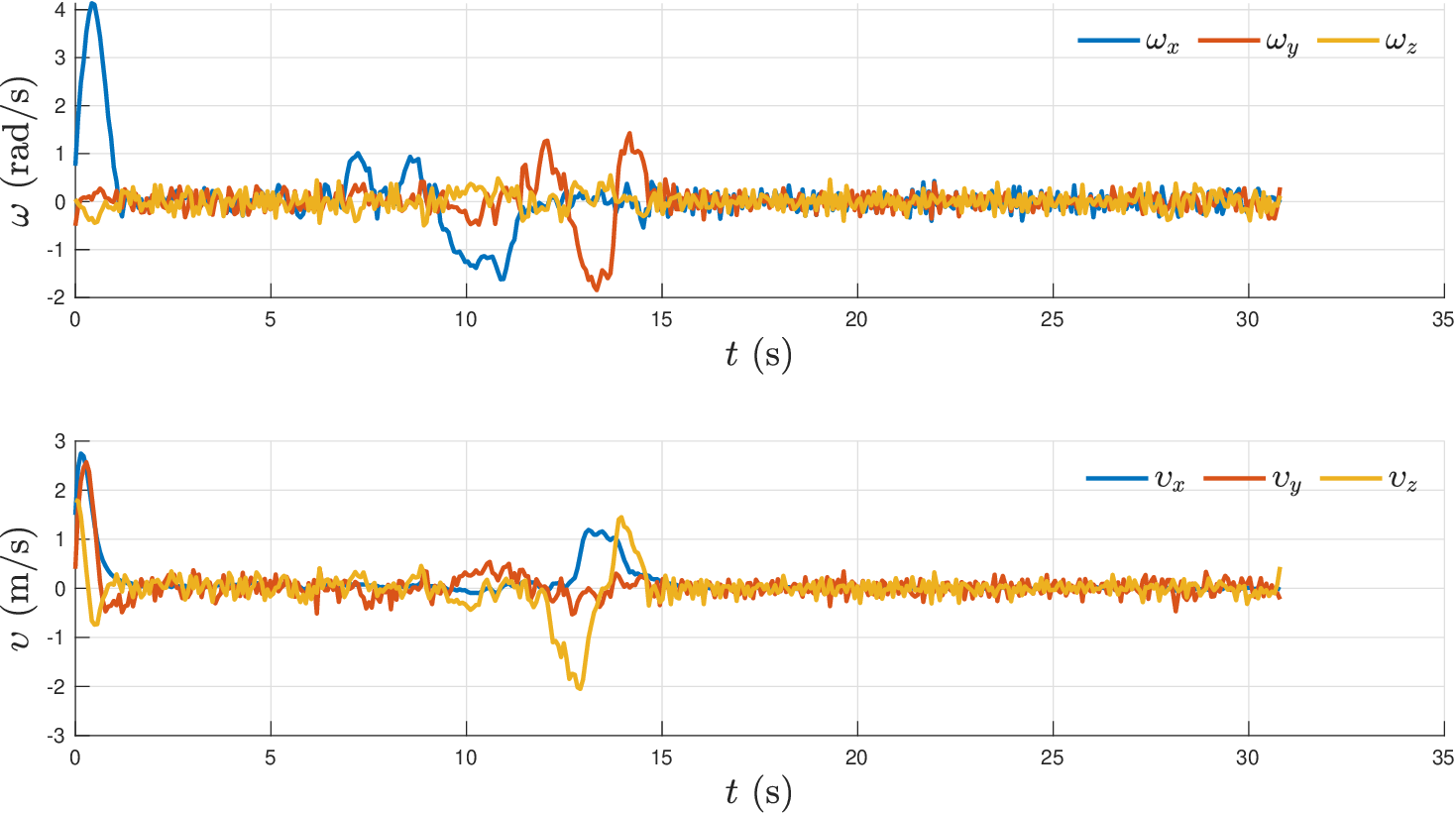}
    \caption{Plots of the angular and translational velocity estimation errors obtained from the experimental evaluation of the FTS-PE scheme.} 
    \label{fig:velocity estimation errors experimental}
\end{figure}
In Fig. \ref{fig:velocity estimation errors experimental}, estimation errors in the angular and translational velocity obtained by the FTS-PE from the Zed2i sensor's point cloud and angular velocity measurements, are seen to converge to small neighborhoods of 
zero errors in a short period of time. A large change in attitude and relatively smaller change in position of the sensor is made between 7 to 14 s after data gathering starts, which leads to sudden changes in the estimation errors. However, these 
errors subsequently decrease within 2 s to small neighborhoods of zero errors. This qualitative feature is seen also in Fig. \ref{fig:pose estimation errors experimental}, which shows that the norm of the position estimation error, 
starting from non-zero initial conditions, converges to nearly zero within short time periods initially, and after a large change in motion is made around 7 s into the data gathering 
period. This is also demonstrated in the plot of the attitude estimation error, parameterized by the principal rotation angle of the attitude estimation error matrix.

\section{Conclusion}\label{sec: conc}
\vspace*{-1mm}
This paper presents theoretical, computational and experimental results for a novel nonlinear finite-time stable pose estimation scheme for rigid bodies. This scheme utilizes a set of vector measurements in 
the form of point clouds given by an onboard sensor along with angular velocity measurements from an inertial measurement unit, to estimate pose, angular and translational velocities. It can also be 
implemented with additional sensors to mitigate the effects of measurement noise and sensor failure through sensor redundancy. The FTS pose estimator is H\"{o}lder-continuous and model-free, and is designed directly on the 
Lie group of rigid body motions, $\SE$. It is shown that the proposed estimator is almost globally finite-time stable from almost all initial conditions except those in a set of 
zero measure on $\Ta\SE$. Finite-time stability guarantees a faster convergence of estimation errors in finite time, as well as additional robustness to measurement noise compared to 
asymptotically or exponentially stable schemes. This work analyzes both finite-time stability and robustness of the FTS pose estimator by using an energy-like Morse-Lyapunov function on $\Ta\SE$, 
which encapsulates the energy contained in the state estimation errors. The estimation scheme is discretized using a geometric variational integration scheme that preserves the geometry of the state space. Results from numerical 
simulations and an experiment, validate the performance of the proposed estimator by demonstrating its finite-time convergence and robustness properties. Moreover, the behavior of this estimation scheme is compared with two state-of-the-art filters for pose estimation in 
numerical simulations. These simulations show that the FTS-PE and VPE, unlike the DQMEKF filter, are stable and effective in filtering out noise.  It is shown that our proposed FTS-PE scheme achieves 
finite-time stable convergence of estimation errors to $(I,0)$, and in the presence of measurement noise, it converges to a bounded neighborhood of $(I,0)$ in finite time. Future 
research will consider directly designing a FTS pose estimation scheme in discrete time from multi-rate measurements obtained from point cloud and inertial sensor measurements. In addition, we plan to compare the performance of the FTS-PE 
with a few other pose estimation schemes in experiments.
\begin{ack}                               
The authors acknowledge support from the National Science Foundation awards 2132799 and 2343062. They also acknowledge Curtis Cline, a senior undergraduate student who helped with the experiments conducted using the Zed 2i sensor. 
\end{ack}

\bibliographystyle{NameModel}
\bibliography{ref_automatica}           



\end{document}